\documentclass[a4paper,english,numberwithinsect,cleveref]{lipics-v2021}

\usepackage{multicol}
\usepackage{algorithm}
\usepackage{bm}
\usepackage{algpseudocode}
\usepackage{tikz,tikz-qtree}
\usetikzlibrary{positioning}
\usepackage[bibliography=common]{apxproof}
\nolinenumbers

\usepackage[
    type={CC},
    modifier={by},
    version={4.0},
]{doclicense}

\title{Dynamic Membership for Regular Tree Languages}
\author{Antoine Amarilli}{Univ. Lille, Inria, CNRS, Centrale Lille, UMR 9189 CRIStAL, France \and \url{https://a3nm.net/}}{a3nm@a3nm.net}{https://orcid.org/0000-0002-7977-4441}{
Partially supported by the ANR project EQUUS ANR-19-CE48-0019, by the
Deutsche Forschungsgemeinschaft (DFG, German Research Foundation) – 431183758,
and by the ANR project ANR-18-CE23-0003-02 (“CQFD”).}
\author{Corentin Barloy}{Ruhr University Bochum, Germany \and \url{https://barloy.name/}}{corentin.barloy@rub.de}{https://orcid.org/0000-0001-5420-8761}{
Partially supported by the Deutsche Forschungsgemeinschaft (DFG, German Research Foundation), grant 532727578.}
\author{Louis Jachiet}{Télécom Paris, Institut Polytechnique de Paris, France \and \url{https://louis.jachiet.com/}}{research@jachiet.com}{https://orcid.org/0000-0002-8277-6552}{}
\author{Charles Paperman}{Univ. Lille, Inria, CNRS, Centrale Lille, UMR 9189 CRIStAL, France \and \url{https://paperman.name/}}{charles.paperman@univ-lille.fr}{https://orcid.org/0000-0002-6658-5238}{}
\date{}

\authorrunning{A.~Amarilli, C.~Barloy, L.~Jachiet, C.~Paperman}
\Copyright{A.~Amarilli, C.~Barloy, L.~Jachiet, C.~Paperman}


\nolinenumbers
\hypersetup{
    colorlinks,
    linkcolor={red!50!black},
    citecolor={blue!50!black},
    urlcolor={blue!30!black}
  }

  \usepackage{mathtools}

\newcommand{\NN}{\mathbb{N}}

\newtheoremrep{theorem}{Theorem}[section]       %
\newtheoremrep{proposition}[theorem]{Proposition}
\newtheoremrep{lemma}[theorem]{Lemma}
\newtheoremrep{example}[theorem]{Example}
\newtheoremrep{definition}[theorem]{Definition}
\newtheoremrep{corollary}[theorem]{Corollary}
\newtheoremrep{conjecture}[theorem]{Conjecture}
\newtheoremrep{claim}[theorem]{Claim}
\newtheoremrep{remark}[theorem]{Remark}
\newtheoremrep{fact}[theorem]{Fact}

\AddToHook{env/proposition/begin}{\crefalias{theorem}{proposition}}
\AddToHook{env/lemma/begin}{\crefalias{theorem}{lemma}}
\AddToHook{env/example/begin}{\crefalias{theorem}{example}}
\AddToHook{env/definition/begin}{\crefalias{theorem}{definition}}
\AddToHook{env/corollary/begin}{\crefalias{theorem}{corollary}}
\AddToHook{env/conjecture/begin}{\crefalias{theorem}{conjecture}}
\AddToHook{env/claim/begin}{\crefalias{theorem}{claim}}
\AddToHook{env/remark/begin}{\crefalias{theorem}{remark}}
\AddToHook{env/fact/begin}{\crefalias{theorem}{fact}}

\newcommand{\muH}{\mu_{\mathrm{H}}}
\newcommand{\muV}{\mu_{\mathrm{V}}}
\newcommand{\odotVV}{\odot_{\mathrm{VV}}}
\newcommand{\odotVH}{\odot_{\mathrm{VH}}}
\newcommand{\oplusHH}{\oplus_{\mathrm{HH}}}
\newcommand{\oplusVH}{\oplus_{\mathrm{VH}}}
\newcommand{\oplusHV}{\oplus_{\mathrm{HV}}}
\newcommand{\SigmaV}{\Sigma^{\mathrm{V}}}
\newcommand{\SigmaH}{\Sigma^{\mathrm{H}}}

\newcommand{\LL}{L}

\newcommand{\bigO}{O}

\newcommand{\N}{\mathbb{N}}

\newcommand{\myparagraph}[1]{\subparagraph*{#1.}}

\newcommand{\cF}{F}
\newcommand{\cT}{F}
\newcommand{\calS}{\mathcal{S}}

\newcommand{\SigmaR}{\Sigma_{-}} %
\newcommand{\SigmaF}{\Sigma_{+}}

\newcommand{\fc}{\mathit{fc}}
\newcommand{\ns}{\mathit{ns}}

\begin{document}

\maketitle

\hideLIPIcs

\begin{abstract}
  We study the \emph{dynamic membership problem} for regular tree languages
  under relabeling updates: we
  fix an alphabet $\Sigma$ and a regular tree language $L$ over~$\Sigma$
  (expressed, e.g., as a tree automaton), we are given a tree $T$ with labels
  in~$\Sigma$, and we must maintain the information of whether the tree $T$
  belongs to~$L$ while handling relabeling updates that change the labels
  of individual nodes in~$T$.

  Our first contribution is to show that this problem admits an $O(\log n / \log
  \log n)$ algorithm for any fixed regular tree language, improving 
  over known $O(\log n)$ algorithms. This generalizes
  the known $O(\log n / \log \log n)$ upper bound
  over words, and it matches the lower bound of $\Omega(\log n / \log \log n)$
  from dynamic membership to some word languages and from the existential
  marked ancestor problem. 

  Our second contribution is to introduce a class of regular languages, dubbed
  \emph{almost-commutative} tree languages, and show that dynamic membership to
  such languages under relabeling updates can be decided in constant time per
  update. Almost-commutative languages generalize both commutative languages and
  finite languages: they are the analogue for trees of the \emph{ZG
  languages} enjoying constant-time dynamic membership over words. Our main
  technical contribution is to show that this class is conditionally optimal
  when we assume that the alphabet features a \emph{neutral letter}, i.e., a letter
  that has no effect on membership to the language. More precisely,
  we show that any regular tree language with a
  neutral letter which is not almost-commutative cannot be maintained in
  constant time under the assumption that the
  prefix-U1 problem from~\cite{amarilli2021dynamic} also does not admit a
  constant-time algorithm.
\end{abstract}

\section{Introduction}
\label{sec:intro}
\emph{Regular tree languages}, and the corresponding
\emph{tree automata}~\cite{comon2008tree}, are a way to enforce
structural constraints on trees while generalizing regular languages
over words. For any fixed tree language $L$ represented as a tree automaton $A$, given
an input tree $T$ with $n$ nodes, we can verify in $O(n)$ whether $T$ belongs to~$L$, simply
by running $A$ over~$T$. However, in some settings, the tree~$T$
may be dynamically updated, and we want to efficiently maintain the
information of whether $T$ belongs to the language. This problem has been
studied in the setting of XML documents and schemas under the name of
\emph{incremental schema
validation}~\cite{balmin2004incremental}.
In this paper we study the theory of this problem in the RAM model with unit
cost and logarithmic word size.
We 
call the problem \emph{dynamic membership} in line with earlier work on regular word
languages~\cite{amarilli2021dynamic}.
We focus on the \emph{computational complexity} of maintaining membership, as
opposed, e.g., to \emph{dynamic
complexity}~\cite{gelade2012dynamic,tschirbs2023dynamic,schmidt2023work}, which studies the expressive
power required (in the spirit of descriptive complexity).
We further focus on the measure of \emph{data complexity}~\cite{vardi1982complexity},
i.e., we assume
that the regular language is fixed, and measure complexity only as a
function of the input tree.
We study efficient algorithms 
for the dynamic membership problem, and (when possible) matching lower bounds.

The naive algorithm for dynamic membership is to re-run $A$ over~$T$ after each
change, which takes $O(n)$. However, more efficient algorithms are
already known. Balmin et al.~\cite{balmin2004incremental} show that
we can maintain membership to
any fixed regular tree language in time $O(\log^2 n)$ under leaf
insertions, node deletions, and substitutions (also known as node relabelings). We focus in this
work on the specific case of relabeling updates, and it is
then known that the complexity can be lowered
to $O(\log n)$: see, e.g., \cite{amarilli2018enumeration} which relies on
tree decomposition balancing techniques~\cite{bodlaender1998parallel}. This is
slightly less favorable than the $O(\log n / \log \log n)$ complexity known for
dynamic membership for regular word
languages~\cite{frandsen1995dynamic,amarilli2021dynamic}.

In terms of lower bounds, the dynamic membership problem for trees under
relabeling
updates is known to require time $\Omega(\log n / \log \log n)$, already for
some very simple fixed languages. These lower bounds can come from two different
routes. 
One first route is the \emph{existential marked ancestor
problem}, where we must maintain a tree under
relabeling updates that mark and unmark nodes, and
efficiently answer queries asking whether a given node
has a marked ancestor. This problem is known to admit
an unconditional $\Omega(\log n /\log \log n)$ lower
bound in the cell probe model~\cite{alstrup1998marked}, and it turns out that
both the updates and the queries can 
be phrased as relabeling updates for a fixed regular tree language.
One second route is 
via the dynamic membership problem in the more restricted setting of word
languages, where there are known $\Omega(\log n /\log \log n)$ lower
bounds for some languages~\cite{skovbjerg1997dynamic,amarilli2021dynamic}.

These known results leave a small complexity gap between $O(\log n)$ and
$\Omega(\log n / \log \log n)$.
Our first contribution in this work is to address this gap, by
presenting an algorithm for the dynamic membership problem 
that achieves complexity $O(\log n /\log \log n)$,
for any fixed regular tree language, and after linear-time preprocessing of the input
tree. This matches the known lower bounds, and
generalizes the known algorithms for dynamic membership to word languages while
achieving the same complexity. Our algorithm 
iteratively processes the tree to merge sets of nodes into clusters and
recursively process the tree of clusters; it is very related to the tree
contraction technique of Gazit et al.~\cite{gazit1988optimal} and
is reminiscent of other algorithms such as the tree contraction technique of
Miller and Reif~\cite{miller1985parallel}, see
also~\cite{bodlaender1998parallel}.
More precisely, our algorithm regroups tree nodes into clusters which have
logarithmic size, ensuring that
clusters fit into a memory word. This ensures that the effects of updates on clusters can be
handled in constant time by the standard technique of tabulating during the preprocessing. Note that clusters may contain holes, i.e., they amount to contexts,
which we handle using the notion of forest algebras.
Then, the algorithm considers the induced tree structure over clusters, and
processes it
recursively, ensuring that the tree size is divided by $\Theta(\log n)$ at each
step. The main challenge in the algorithm is to ensure that a suitable 
clustering of the trees at each level can be computed in linear time.

Having classified the overall complexity of dynamic membership, the next
question is to refine the study and show language-dependent bounds. Indeed,
there are restricted families of regular tree languages over which we can beat
$O(\log n / \log \log n)$. For instance, it is easy to achieve $O(1)$ time per
update when the tree language is finite, or when it is commutative (i.e., only
depends on the number of label occurrences and not the tree structure). What is
more, in the
setting of word languages, the complexity of dynamic membership is in $\Theta(\log
n / \log \log n)$ only for a restricted subset of the word languages, 
namely, those outside the class QSG defined in~\cite{amarilli2021dynamic} (see
also~\cite{frandsen1995dynamic}). Some word languages admit constant-time dynamic
membership algorithms: they were classified in~\cite{amarilli2021dynamic}
conditional to the hypothesis that there is no data structure achieving constant
time per operation for a certain problem, dubbed
\emph{prefix-$U_1$} and essentially amounting to maintaining a weak form of
priority queue.

Our second contribution is to introduce a class of
regular tree languages
for which dynamic membership under relabeling updates can be achieved in
constant time per update, generalizing commutative tree languages and finite
tree languages.
Specifically, we define the class of \emph{almost-commutative} tree
languages: these are obtained as the Boolean closure of regular tree languages
which impose a commutative condition on the so-called \emph{frequent} letters
(those with more occurrences than a constant threshold) and imposing that the projection
to the other letters (called \emph{rare}) forms a specific constant-sized tree.
We show that it is decidable whether a regular tree language is
almost-commutative
(when it is given, e.g., by a tree automaton). 
The motivation for almost-commutative languages is that we can show that dynamic
membership to such
languages can be maintained in $O(1)$ under relabeling updates, generalizing the
$O(1)$ algorithm of~\cite{amarilli2021dynamic} for ZG monoids. 

Our third contribution is to show that, conditionally, the almost-commutative tree languages in fact
\emph{characterize} those regular tree languages that enjoy constant-time dynamic
membership under relabeling updates,
when we assume that the language
features a so-called \emph{neutral letter}.
Intuitively, a letter $e$ is \emph{neutral} for a language $L$ if it is invisible in
the sense that nodes labeled with $e$ can be replaced by the forest of their
children without affecting membership to~$L$.
Neutral letters are often assumed
in the setting of word
languages~\cite{koucky2005bounded,barrington2005first}, and they make
dynamic membership for word languages collapse to the same problem
for their syntactic monoid (as was originally studied
in~\cite{skovbjerg1997dynamic}, and in~\cite{amarilli2021dynamic} as a first
step). 

When focusing on regular tree languages with a neutral letter, we show that for
\emph{any} such language which is \emph{not} almost-commutative, dynamic membership
cannot be maintained in constant time under relabeling updates,
subject to the hypothesis from~\cite{amarilli2021dynamic} that the prefix-$U_1$
problem does not admit a constant-time data structure. Thus, the
$O(1)$-maintainable regular tree languages with a neutral letter can be
characterized conditionally to the same hypothesis as in the case of words. 
We show this conditional lower bound via an (unconditional) algebraic characterization of the class of
almost-commutative languages with neutral letters: they are precisely those
languages with syntactic forest algebras whose vertical monoid satisfies the
equation ZG ($x^{\omega+1} y = y x^{\omega+1}$, i.e., group elements are
central), which was also the (conditional) characterization for
$O(1)$-maintainable word languages with a neutral 
letter~\cite{amarilli2021dynamic}.
The main technical challenge to show
this result is to establish that every tree language with a ZG vertical monoid
must fall in our class of almost-commutative languages: this uses a
characterization of ZG tree languages via an equivalence relation analogous to
the case of ZG on words~\cite{amarilli2023locality}, but with more challenging
proofs because of the tree structure.

\myparagraph{Paper structure}
We give preliminaries in \cref{sec:prelim} and introduce the dynamic membership
problem. %
We then give in \cref{sec:forest} some further preliminaries about
\emph{forest algebras}, which are instrumental to the proof of our results.
We then present our first contribution
in \cref{sec:lll}, namely, an $O(\log n / \log \log n)$ algorithm for
dynamic membership to any fixed regular tree language. 
We then introduce almost-commutative languages in \cref{sec:zgupper} and
show that dynamic membership to these
languages is in $O(1)$. In \cref{sec:zglower} we show our matching lower bound:
in the presence of a neutral letter, and assuming that prefix-$U_1$ cannot be
solved in constant time, then the dynamic membership problem cannot be solved in constant
time for any non-almost-commutative regular language. We conclude 
in \cref{sec:extension}.

Because of space limitations, detailed proofs are deferred to the appendix. Note
that a version of the results of
\Cref{sec:zgupper,sec:zglower} was presented in~\cite{barloy2024complexity} but
never formally published.

\section{Preliminaries}
\label{sec:prelim}
\subparagraph*{Forests, trees, contexts.}
We consider ordered forests of rooted, unranked, ordered trees on a finite
alphabet $\Sigma$, and simply call them \emph{forests}.
Formally, we define \emph{$\Sigma$-forests} and \emph{$\Sigma$-trees} by mutual
induction: a $\Sigma$-forest is a (possibly empty) ordered sequence of
$\Sigma$-trees, and a $\Sigma$-tree consists of a \emph{root node} with a label
in $\Sigma$ and with a $\Sigma$-forest of child nodes.

We assume familiarity with standard terminology on trees
including parent node, children node, ancestors, descendants, root nodes, 
siblings, sibling sets, previous sibling, next sibling, prefix order, etc.; see~\cite{comon2008tree} for details. As we work with ordered
forests, we will always consider root nodes to be siblings, and define the
previous sibling and next sibling relationships to also apply to roots.
We will often abuse notation and identify forests with their sets of nodes,
e.g., for a forest $F$, we define functions on $F$ to mean functions defined on
the nodes of $F$.
We say that two forests have the \emph{same
shape} if they are identical up to changing the labels of their nodes (which are
elements of~$\Sigma$).
The \emph{size} $|F|$ of a $\Sigma$-forest $F$ is its number of nodes, and we
write $|F|_a$ for $a \in \Sigma$ to denote the number of occurrences of~$a$
in~$F$.

\subparagraph*{Forest languages.}
A \emph{forest language} $L$ over an alphabet~$\Sigma$
is a subset of $\Sigma$-forests, and a \emph{tree
language} is a subset of $\Sigma$-trees:
throughout the paper we study forest
languages, but of course our results extend to tree languages as well.
We write $\overline{L}$ for the complement of~$L$.
We will be specifically interested in \emph{regular forest
languages}: among many other equivalent characterizations~\cite{comon2008tree}, 
these can be formalized via finite automata, which we formally introduce below
(first for words and then for forests).

A \emph{deterministic word automaton} over a set $S$ is a tuple $A = (Q, F, q_0, \delta)$ where
$Q$ is a finite set of \emph{states}, $F \subseteq Q$ is a subset of \emph{final
states}, $q_0 \in Q$ is the \emph{initial state}, and $\delta \colon Q \times S
\to Q$ is a transition function. For $w = s_1 \cdots s_k$ a word over~$S$, the
\emph{state reached by~$A$ when reading $w$} is defined inductively as usual: if
$w = \epsilon$ is the empty word then the state reached is~$q_0$, otherwise the
state reached is $\delta(q, s_k)$ with $q$ the state reached when reading $s_1
\cdots s_{k-1}$. The word $w$ is \emph{accepted} if the state reached by~$A$
when reading~$w$ is in~$F$.

We define \emph{forest automata} over forests, which recognize forest languages
by running a word automaton over the sibling sets.
These automata are analogous to the \emph{hedge automata} of~\cite{comon2008tree} (where
horizontal transitions are specified by regular languages), and to
the forest automata of~\cite{bojanczyk2008forest} (where horizontal transitions
are specified by a monoid).
Formally, a \emph{forest automaton} over~$\Sigma$ is a tuple $A = (Q, A', \delta)$ where $Q$ is a
finite set of states,  
$A' = (Q', F', q_0', \delta')$ is a word automaton over the set~$Q$, and
$\delta\colon Q' \times \Sigma \to
Q$ is a \emph{vertical transition function} which associates to every state
$q' \in Q'$ of~$A'$ and letter $a \in \Sigma$ a state $\delta(q',a) \in Q$.
The \emph{run} of~$A = (Q, A', \delta)$ on a $\Sigma$-tree $T$ is a function
$\rho$ from~$T$ to~$Q$ defined inductively as follows:
\begin{itemize}
\item For $u$ a leaf of~$T$ with label $a \in \Sigma$,
we have $\rho(u) = \delta(q_0', a)$ for $q_0'$ the initial state of~$A'$;
\item For $u$ an internal node of~$T$ with label $a \in \Sigma$,
letting $u_1, \ldots, u_k$ be its successive
children, and letting $q_i \coloneq \rho(u_i)$ for all $1 \leq
i \leq k$ be the inductively computed images of the run, 
    for $q$ the state reached in~$A'$ when reading the word
$q_1 \cdots q_k$, we let $\rho(u) \coloneq \delta(q, a)$.
\end{itemize}
We say that a $\Sigma$-forest $F$ is \emph{accepted} by the forest automaton $A$ if, letting $T_1, \ldots,
T_k$ be the trees of~$F$ in order, and letting $q_1, \ldots, q_k$ be the states
to which the respective roots of the $T_1, \ldots, T_k$ are mapped by the
respective runs of~$A$, then the word $q_1 \cdots q_k$ is accepted by the
word automaton~$A'$.
(In particular, the empty forest is accepted iff $q_0' \in F'$.)
The \emph{language} of~$A$ is the set of~$\Sigma$-forests that it accepts.

\subparagraph*{Relabelings, incremental maintenance.}
We consider \emph{relabeling
updates} on forests $F$, that change node labels without changing the tree shape.
A \emph{relabeling update} gives a node $u$ of $F$ (e.g., as a pointer to
that node) and a label $a \in \Sigma$, and changes the label of the
node $u$ to~$a$.

We fix a regular language $L$ of $\Sigma$-forests given by a fixed forest automaton
$A$, and we are given as input a $\Sigma$-forest $F$. We can compute in linear time
whether $A$ accepts $F$. The \emph{dynamic membership problem} for~$L$ is
the task of maintaining whether the current forest belongs to~$L$, under
relabeling
updates. We study the complexity of this problem, defined as the worst-case
running time of an algorithm after each update, expressed as a function of the
size $|F|$ of~$F$. Note that the language~$L$ is assumed to be fixed, and
is not accounted for in the complexity. We work in the RAM model with unit cost
and logarithmic word size, and consider dynamic membership algorithms that run a
linear-time \emph{preprocessing} on the input forest~$F$ to build data structures
used when processing updates. (By contrast, our lower bound results will hold
without the assumption that the preprocessing is in linear time.)

As we reviewed in the introduction, for any fixed regular forest language, it is
folklore that the dynamic membership problem on a tree with $n$ nodes can be solved 
under relabeling updates in time $O(\log n)$ per
update: see for instance \cite{amarilli2018enumeration}. Further, for some
forest languages, some unconditional lower bounds 
in $\Omega(\log n / \log \log n)$ are known. This includes some forest languages defined from
intractable word languages, for instance the forest language $L_1$
on alphabet $\Sigma = \{0, 1, \#\}$ where the word on $\Sigma^*$ formed by the
leaves in prefix order is required to fall in the word language $L_2 \coloneq
(0^*10^*10^*)^* \#\Sigma^*$,
i.e., a unique $\#$ with an even number of $1$'s before it. Indeed, dynamic
membership to $L_2$ (under letter substitutions) admits
an $\Omega(\log n / \log \log n)$ lower bound from the
\emph{prefix-$\mathbb{Z}_2$} problem (see \cite[Theorem 3]{fredman1989cell},
and~\cite{frandsen1995dynamic,amarilli2021dynamic}), hence so does $L_1$.
Intractable forest languages also include
languages corresponding to the \emph{existential marked ancestor} problem
\cite{alstrup1998marked}, e.g., the language $L_3$ on alphabet $\Sigma = \{e, m,
\#\}$ where we have a
unique node $u$ labeled $\#$ and we require that $u$ has an ancestor labeled~$m$. Indeed,
the existential marked ancestor problem of~\cite{alstrup1998marked} allows us to
mark and unmark nodes over a fixed
tree (corresponding to letters $e$ for unmarked nodes and $m$ for marked nodes),
and allows us to query whether a node (labeled~$\#$) has a marked ancestor.
Thus, the existential marked ancestor problem immediately reduces to dynamic
membership for~$L_3$, which inherits the 
$\Omega(\log n / \log \log n)$ lower bound from~\cite{alstrup1998marked}.

\begin{toappendix}
\section{Algebraic Theory of Regular Forest Languages}
\label{sec:algebra}
In this appendix, we develop the algebraic theory of regular languages of forests in more depth than in~\cref{sec:forest}.

\myparagraph{Forest algebra}
Recall our definition of \emph{forest algebras}: a forest algebra is a two-sorted algebra \((V,H)\) along with operations
\(\oplusHH, \oplusHV, \oplusVH, \odotVV\) and \(\odotVH\); and distinguished (neutral) elements \(\epsilon\) and \(\square\). We have also defined the free forest algebra \(\Sigma^{\nabla}= (\SigmaV, \SigmaH)\).

The product \((V_{1},H_{2})\times (V_{2},H_{2})\) of two forest algebras is defined as \((V_{1}\times V_{2}, H_{1}\times H_{2})\) with all operations applied componentwise.
Also recall that we defined morphisms \(\mu_V, \mu_H\) from a forest algebra
\((V,H)\) to another forest algebra \((V',H')\), which we write $\mu$ by abuse
of notation.
Thanks to morphisms, we can construct new forest algebras from a given forest algebra \((V,H)\).
\begin{itemize}
  \item \emph{Subalgebra.} We say that \((V',H')\) is a \emph{subalgebra} of \((V,H)\) if there exists an injective morphism \(\mu\)
    from \((V',H')\) to \((V,H)\). It means that both \(\muV\) and \(\muH\) have to be injective.
  \item \emph{Quotient.} We say that \((V',H')\) is a \emph{quotient} of \((V,H)\) if there exists a surjective morphism \(\mu\)
    from \((V,H)\) to \((V',H')\). It means that both \(\muV\) and \(\muH\) have to be surjective.
      \item \emph{Division.} We say that \((V',H')\) \emph{divides} \((V,H)\) if it is a quotient of a subalgebra of \((V,H)\).
\end{itemize}

A forest language \(L\) over \(\Sigma\) is \emph{recognized} by a forest algebra
\((V,H)\) if there exists a morphism \(\mu\colon \Sigma^{\nabla}\to (V,H)\) and a subset \(H'\subseteq H\) such that
\(L = \mu^{-1}(H')\).

\myparagraph{Comparison to the original definition}
Forest algebra were introduced in~\cite{bojanczyk2008forest} with a slightly different definition.
In~\cite{bojanczyk2008forest}, a forest algebra is also an algebra \((V,H)\) coming with several operations satisfying axioms:
\begin{itemize}
  \item \(\oplusHH, \odotVV, \odotVH\) as in our framework, satisfying Action
    (composition), Faithfulness, and the monoid axioms.
    (Note that they do not require Action (neutral).)
  \item two operations \(\text{in}_{l}\) and \(\text{in}_{r}\) from \(H\) to \(V\) satisfying the \emph{Insertion} axiom: for every \(h,g\in H\),
        \(\text{in}_{l}(h)\odotVH g = h \oplusHH g\) and
        \(\text{in}_{r}(h)\odotVH g = g \oplusHH h\).
\end{itemize}

We will see that these definitions are indeed equivalent.
That is to say that we can define \(\text{in}_{r}\) and \(\text{in}_{l}\) from \(\oplusHV\) and \(\oplusVH\),
and vice versa.

First assume that we are given a forest algebra \((V,H)\) in the sense defined
in this paper, that is to say with two operations \(\oplusHV\) and \(\oplusVH\).
For \(h\in H\), we define \(\text{in}_{l}(h) = h \oplusHV \square\) and \(\text{in}_{r}(h) =\square \oplusVH h\).
We want to show that the Insertion axiom holds for \(\text{in}_{l}\), as the other case is symmetric.
Let \(h,g\in H\), we have that  \(\text{in}_{l}(h) \odotVH g  = (h\oplusHV
\square) \odotVH g = h \oplusHH (\square\odotVH g) \), by the Mixing axiom.
This last part is equal to \(h \oplusHH g\) thanks to Action (neutral).

Now, assume that we have a forest algebra \((V,H)\) in the sense of~\cite{bojanczyk2008forest}, that is to say with two operations \(\text{in}_{l}\) and \(\text{in}_{r}\).
For \(h\in H\) and \(v\in V\), we define \(h\oplusHV v = \text{in}_{l}(h)
\odotVV v\) and \(v\oplusVH h = \text{in}_{r}(h) \odotVV v\).
We want to show that the Mixing axiom holds for \(\oplusHV\), as the other case is symmetric.
Let \(h,g\in H\) and \(v\in V\), we have \((h\oplusHV v) \odotVH g = (\text{in}_{l}(h) \odotVV v) \odotVH g = \text{in}_{l}(h) \odotVH (v \odotVH g) \) by the Action (composition) axiom.
This is equal to \(h\oplusHH (v\odotVH g)\) by Insertion, which is the desired value.
Finally, we have to see that we can prove Action (neutral).
Let \(h\in H\), note that we can find \(v\in V\) such that \(h=v\odotVH\epsilon\) ($v=\text{in}_{l}(h)$ suffices by Insertion).
Then \(\square \odotVH h = (\square \odotVV v) \odotVH \epsilon = v\odotVH \epsilon = h\).
We used Action (composition) once.

\myparagraph{Syntactic forest algebra}
We define here the syntactic forest algebra of a forest language thanks to an equivalence relation similar to the Myhill-Nerode
relation for monoids~\cite[Definition 3.13]{bojanczyk2008forest}.
Let \(L\) be a forest language over \(\Sigma\).
Its \emph{syntactic relation} \(\sim_{L}\) is a pair of equivalence relations: one over \(\SigmaV\) and one over \(\SigmaH\).
We denote both relations with the same symbol \(\sim_{L}\).
For \(F\) and \(F'\) two \(\Sigma\)-forests, we let \(F\sim_{L}F'\) whenever for any \(\Sigma\)-contexts \(C\), we have that \(C(F')\) and \(C(F')\) are either both in \(L\) or both not in \(L\).
For \(C\) and \(C'\) two \(\Sigma\)-contexts, we let \(C\sim_{L}C'\) whenever for any \(\Sigma\)-context \(R\) and \(\Sigma\)-forest \(F\), we have that \(R(C(F))\) and \(R(C'(F))\) are either both in \(L\) or both not in \(L\).
The \emph{syntactic forest algebra} \((V_{L}, H_{L})\) of \(L\) is the set of equivalence classes of \(\SigmaV\) and \(\SigmaH\) under \(\sim_{L}\).
Any forest algebra operation between two equivalence classes can be defined as the equivalence class of the corresponding operation applied to any representatives.
It can be shown that these operations are well-defined, that is to say that the result does not depend on the chosen representatives~\cite[Lemma 3.12]{bojanczyk2008forest}.
The syntactic morphism \(\mu_{L}\) of \(L\) is the morphism that maps every context/forest to its equivalence class in the syntactic relation.

By taking the subset of \(H_{L}\) that consists of the equivalence classes of forests in \(L\), we immediately have that \((V_{L},H_{L})\) recognizes \(L\).
The surjectivity promised in~\cref{sec:forest} is also easy to obtain: for any element \(v\in V_{L}\), we can take any member \(C\) of \(v\) seen as an equivalence class to have
\(\mu_{L}(C)=v\).
Similarly, the minimality, also promised in~\cref{sec:forest}, follows from the definition of \(\sim_{L}\).
Indeed, by the contrapositive, let \(C\) and \(C'\) be two \(\Sigma\)-contexts.
That \(R(C(F))\) and \(R(C'(F))\) belong or not to \(L\) together for all \(\Sigma\)-contexts \(R\) and \(\Sigma\)-forests \(F\) is precisely the definition of \(C\sim_{L}C'\), and thus
it implies \(\mu_{L}(C) = \mu_{L}(C')\).
Moreover, the syntactic forest algebra is also minimal among the recognizers of
\(L\), as we will now show:

\begin{lemma}[{\cite[Proposition 3.14]{bojanczyk2008forest}}]
  \label{lem:syntactic_divides_all_rec}
  Let \(L\) be a regular forest language.
  Then its syntactic forest algebra divides any other forest algebra recognizing \(L\).
\end{lemma}

\begin{proof}
  The reference in~\cite{bojanczyk2008forest} is not exactly what we claim, but can be easily deduced.
  We give a self-contained proof below for convenience.

  Let \(\mu\colon\Sigma^{\nabla}\to (V,H)\) that recognizes \(L\).
  We assume at first that it is surjective and we will prove that \((V_{L},H_{L})\) is a quotient of \((V,H)\).
  In that case, let \(h\) be an element of \(H\) (the case for \(V\) is identical).
  We claim that any two elements \(F,G\) of \(\mu^{-1}(h)\) are \(\sim_{L}\)-equivalent.
  Indeed, for \(C\) any \(\Sigma\)-context, we have that \(\mu(C(F)) = \mu(C) \odotVH \mu(F)  = \mu(C) \odotVH \mu(G) = \mu(C(G)) \).
  Hence, \(C(F)\) and \(C(G)\) are both in \(L\) if \(\mu(C(F))\) is accepting, and both not in \(L\) otherwise.
  This means that we can build a morphism \(\nu\colon (V,H) \to (V_{L}, H_{L})\) by \(\nu(t)= \mu_{L}(\mu^{-1}(t))\) for \(t\) in \(V\) or \(H\).
  This gives a single value by the preceding observation, and can be checked easily to be a morphism.
  Finally, it is surjective because any \(t\) in \(V_{L}\) or \(H_{L}\) can be lifted into \(\Sigma^{\nabla}\) into some \(X\) such that \(\mu_{L}(X)=t\).
  Thus \(\mu(X)\) gives a value such that \(\nu(\mu(X))=X\).

  If \(\mu\) is not surjective, we notice that it is surjective on its image \(\mu(\Sigma^{\nabla})\), which is a subalgebra of \((V,H)\).
  Thus we can apply the previous part on \(\mu(\Sigma^{\nabla})\) to obtain a division.
\end{proof}

\myparagraph{Boolean operations}
The goal of this paragraph is to show the closure of languages whose syntactic
forest algebra is in ZG under Boolean operations, which we now state
(anticipating the definition of ZG from \cref{def:zg} in \cref{sec:zgupper}):

\begin{claim}
  \label{clm:closure}
  Let $L_1$ and $L_2$ be two regular forest languages whose syntactic forest
  algebras are in ZG. Then the syntactic forest algebras of the intersection $L_1
  \cup L_2$, of the union $L_1 \cap L_2$, and of the complement $\overline{L_1}$ are also in ZG.
\end{claim}

We prove this result in the rest of this section of the appendix.
First, note that \cref{clm:closure} would follow immediately from the more general theory of pseudovarieties that gives strong relations between certain classes
of regular languages and of forest algebras.
It was introduced by Eilenberg~\cite{Eilenberg1976bookB} in the case of finite monoids.
For finite forest algebras, there is no exposition focusing solely on them.
However, Bojańczyk~\cite{bojanczyk2015monads} extended pseudovariety theory to the very broad scope of monads, in which forest algebras fall
(see~\cite[Section 9.3]{bojanczyk2015monads}).
For reference, the results in~\cite{bojanczyk2015monads} that can be used to
prove the desired closure properties are \cite[Theorem
4.2]{bojanczyk2015monads} and \cite[Corollary
4.6]{bojanczyk2015monads}.

However, we choose to prove these results here in a more elementary and self-contained manner, since we do not need the full power of pseudovariety theory.
We start by settling the case of complementation.

\begin{lemma}
  \label{lem:syntactic_compl}
  The regular forest languages \(L\) and \(\overline{L}\) have the same syntactic forest algebras.
\end{lemma}

\begin{proof}
  This comes from the fact that for any two \(\Sigma\)-forests \(F\) and \(G\), we have that
  \(F\) and \(G\) are simultaneously in \(L\) if and only if there are simultaneously in \(\overline{L}\).
  In other words, the condition on~$L$ in the definition of the syntactic forest algebra is invariant under complementation.
  Therefore, \(\sim_{L}\) and \(\sim_{\overline{L}}\) are identical and the result follows.
\end{proof}

We continue with the cases of intersection and union.

\begin{lemma}
  \label{lem:syntactic_inter}
 Let \(L_{1}\) and \(L_{2}\) be two regular forest languages with respective syntactic forest algebras \((V_{L_{1}},H_{L_{1}})\) and \((V_{L_{2}},H_{L_{2}})\).
 Then the syntactic forest algebra of \(L_{1}\cap L_{2}\) (resp. \(L_{1}\cup L_{2}\)) divides the product \((V_{L_{1}},H_{L_{1}}) \times (V_{L_{2}},H_{L_{2}})\).
\end{lemma}

\begin{proof}
  The first step is to prove that \(L_{1}\cap L_{2}\) (resp. \(L_{1}\cup L_{2}\)) is recognized by \((V_{L_{1}},H_{L_{1}}) \times (V_{L_{2}},H_{L_{2}})\).
  We can then apply~\cref{lem:syntactic_divides_all_rec} to conclude.

  We take the two syntactic morphisms \(\mu_{L_{1}}\) and \(\mu_{L_{2}}\) and
  construct their product \(\nu\colon \Sigma^{\nabla}\to (V_{L_{1}},H_{L_{1}}) \times (V_{L_{2}},H_{L_{2}})\).
  It is defined by associating to a \(\Sigma\)-context/\(\Sigma\)-forest \(T\) the value \((\mu_{L_{1}}(T), \mu_{L_{2}}(T))\).
  Let \(H'_{1}\) and \(H'_{2}\) be the respective subsets of \(H_{1}\) and \(H_{2}\) that recognize \(L_{1}\) and \(L_{2}\).
  It is straightforward that \(L_{1}\cap L_{2}\) is recognized by \(\nu\) and the subset \(H'_{1}\times H'_{2}\).
  Similarly, \(L_{1}\cup L_{2}\) is recognized by \(\nu\) and the subset \((H'_{1}\times H_{2}) \cup (H_{1}\times H'_{2})\).
\end{proof}

Finally, in light of~\cref{lem:syntactic_inter} and~\cref{lem:syntactic_compl},
we only need to prove that being in ZG is preserved under product and division to prove~\cref{clm:closure}.
Note that this would fall once again in the framework of pseudovarieties,
and that the statement is true for any class defined thanks to equations, for a precise notion of equations.
See the book of Pin~\cite[Chapter XII]{pinMPRI} for a presentation of the theory for finite monoids,
or the article of Bojańczyk~\cite[Section 11]{bojanczyk2015monads} (in particular Theorem~11.3) for a presentation for monads (which encompass forest algebras).

\begin{lemma}
  Let \((V_{1},H_{1})\) and \((V_{2},H_{2})\) be forest algebras.
  \begin{itemize}
    \item If both are in ZG, then \((V_{1},H_{1})\times (V_{2},H_{2})\) is in ZG.
    \item If \((V_{1},H_{1})\) is in ZG and \((V_{2},H_{2})\) divides \((V_{1},H_{1})\), then \((V_{2},H_{2})\) is in ZG.
  \end{itemize}
\end{lemma}

\begin{proof}
  For division, we will prove the statement separately in the special cases of subalgebras and quotients.
  The result follows from the definition of division.
  \begin{itemize}
    \item \emph{Product.} Assume \((V_{1},H_{1})\) and \((V_{2},H_{2})\) are in ZG.
          Let \((v_{1},v_{2})\) and \((w_{1},w_{2})\) be in \(V_{1}\times V_{2}\).
          We need a small direct result about the idempotent power in a product.
          Suppose \((v_{1},v_{2})^{\omega} = (x_{1},x_{2})\) for some \(x_{1},x_{2}\).
          By idempotence, \((x_{1}^{2}, x_{2}^{2}) = (x_{1},x_{2})\).
          Thus \(x_{1}\) (resp.\ \(x_{2}\)) is a power of \(x_{1}\) (resp.\ \(x_{2}\)) that is idempotent.
          Hence \((v_{1},v_{2})^{\omega} = (v_{1}^{\omega},v_{2}^{\omega})\).
          Finally, \((v_{1},v_{2})^{\omega}(w_{1},w_{2}) =(v_{1}^{\omega}w_{1},v_{2}^{\omega}w_{2}) =(w_{1}v_{1}^{\omega},w_{2}v_{2}^{\omega}) = (w_{1},w_{2}) (v_{1},v_{2})^{\omega}   \),
          where we used that \(V_{1}\) and \(V_{2}\) satisfy the equation of ZG.
          This shows that the product also satisfies the equation of ZG.
    \item \emph{Subalgebra.} Assume \((V_{1},H_{1})\) is in ZG and there is an
      injective morphism \(\mu\colon (V_{2},H_{2}) \to (V_{1},H_{1})\).
          Let \(v,w\in V_{2}\).
          We first have to verify that a morphism is well-behaved with regards to the idempotent power,
          that is to say that \( \mu(v)^{\omega} = \mu(v^{\omega})\).
          This is indeed the case as \(\mu(v^{\omega})\) is a power of \(\mu(v)\) that is an idempotent.
          Then \(\mu(v^{\omega+1}w) = \mu(v)^{\omega+1}\mu(w) = \mu(w)\mu(v)^{\omega+1} =\mu(wv^{\omega+1})   \), where
          we used that \(V_{1}\) satisfy the equation of ZG.
          By injectivity, this implies that \(v^{\omega+1}w=wv^{\omega+1}\).
          This shows that the subalgebra also satisfies the equation of ZG.
    \item \emph{Quotient.} Assume \((V_{1},H_{1})\) is in ZG and there is a
      surjective morphism \(\mu\colon (V_{1},H_{1}) \to (V_{2},H_{2})\).
          Let \(v,w \in V_{2}\).
          By surjectivity, we can find \(v',w'\in V_{1}\) such that \(\mu(v')=v\) and \(\mu(w')=w\).
          Using the same observation on idempotents as for subalgebras, \(v^{\omega+1}w = \mu(v'^{\omega+1}w')= \mu(w'v'^{\omega+1})= wv^{\omega+1}\),
          where we used that \(V_{1}\) satisfy the equation of ZG.
          This shows that the quotient also satisfies the equation of ZG. \qedhere
  \end{itemize}
\end{proof}

\myparagraph{Examples}
We give here two examples of forest languages over \(\Sigma=\{a,b\}\) and compute their syntactic forest algebra.
Recall that the syntactic forest algebra is defined as a set of equivalence
classes, so we denote the elements of the vertical and horizontal monoids by a
representative of these equivalence classes.
We conclude the section by giving an example of evaluation of a \((V,H)\)-forest (see definition in~\cref{sec:forest}).

\newcommand{\scalen}{0.8}
\newsavebox{\boxtreea}
\sbox{\boxtreea}{%
\begin{tikzpicture}[-,line width=0.2pt,scale=\scalen, shape=circle, inner sep=0.00cm, level distance=12pt, sibling distance=20pt,every node/.style={transform shape}]
  \node {a}
  child {node {\tiny \(\square\)}};
\end{tikzpicture}%
}
\newcommand{\treea}{\raisebox{-5pt}{\usebox{\boxtreea}}}

Let \(L_{1}\) be the language of \(\Sigma\)-forests with an even number of \(a\).
It is easy to check that two contexts or forests are equivalent for the
syntactic relation if and only if they have the numbers of nodes labeled $a$
have the same parity.
Therefore, the syntactic forest algebra of \(L_{1}\) is \((V_{1},H_{1})\) with \(H_{1}=\{\epsilon, a\}\) and \(V_{1}=\left\{\square, \treea\right\}\).
The tables for each operation are given in~\cref{fig:syntactic_1}.
We omit the table of \(\oplusHV\), as it is identical to the table of
\(\oplusVH\) whose rows and columns are swapped.
Note that \(V_{1}\) and \(H_{1}\) are both isomorphic to the group \(\mathbb{Z}/2\mathbb{Z}\) that computes the addition modulo \(2\).
Hence \(V_{1}\) is commutative, which implies that \(L_{1}\) belongs to ZG (see~\cref{def:zg}).

\begin{figure}
  \centering

  \begin{minipage}{0.24\textwidth}
    \centering
    \[
      \begin{array}{c|cc}
        \oplusHH & \epsilon & a  \\
        \hline
        \epsilon & \epsilon & a   \\
        a & a & \epsilon
      \end{array}
    \]
  \end{minipage}
  \begin{minipage}{0.24\textwidth}
    \centering
    \[
      \begin{array}{c|cc}
        \odotVV & \square & \treea  \\
        \hline
        \square & \square & \treea   \\
        \treea & \treea & \square
      \end{array}
    \]
  \end{minipage}
  \begin{minipage}{0.24\textwidth}
    \centering
    \[
      \begin{array}{c|cc}
        \odotVH & \epsilon & a  \\
        \hline
        \square & \epsilon & a   \\
        \treea & a & \epsilon
      \end{array}
    \]
  \end{minipage}
  \begin{minipage}{0.24\textwidth}
    \centering
    \[
      \begin{array}{c|cc}
        \oplusVH & \epsilon & a  \\
        \hline
        \square & \square & \treea   \\
        \treea & \treea & \square
      \end{array}
    \]
  \end{minipage}

  \caption{Operation tables of the syntactic forest algebra of $L_1$}
  \label{fig:syntactic_1}
\end{figure}

\newsavebox{\boxtreeaa}
\sbox{\boxtreeaa}{%
\begin{tikzpicture}[-,line width=0.2pt,scale=\scalen, shape=circle, inner sep=0.00cm, level distance=12pt, sibling distance=20pt,every node/.style={transform shape}]
  \node {a}
  child {node {a}};
\end{tikzpicture}%
}
\newcommand{\treeaa}{\raisebox{-5pt}{\usebox{\boxtreeaa}}}

\newsavebox{\boxtreeaas}
\sbox{\boxtreeaas}{%
\begin{tikzpicture}[-, line width=0.2pt,scale=\scalen, shape=circle, inner sep=0.00cm, level distance=12pt, sibling distance=20pt,every node/.style={transform shape}]
  \node {a}
  child {node {a} child {node {\tiny \(\square\)}}};
\end{tikzpicture}%
}
\newcommand{\treeaas}{\raisebox{-10pt}{\usebox{\boxtreeaas}}}

\newsavebox{\boxtreeap}
\sbox{\boxtreeap}{%
  \begin{tikzpicture}[ shape=circle, inner sep=0.00cm, level distance=12pt, sibling distance=20pt]
   \node (a1) {a};
   \node (a2) [right=0.1pt of a1] {\tiny +};
   \node (a3) [right=0.1pt of a2] {\tiny \(\square\)};
 \end{tikzpicture}%
}
\newcommand{\treeap}{\usebox{\boxtreeap}}

Let \(L_{2}\) be the language of \(\Sigma\)-forests whose nodes labelled by $a$
are pairwise incomparable, i.e., form an antichain. In other words, the input
forest must not contain two nodes $n$ and $n'$ labeled $a$ such that $n$ is a
strict ancestor of~$n'$. Note that this language is also discussed
in~\cite[Example~7.25]{barloy2024complexity}.
Once again, we can describe all the equivalence classes of the syntactic equivalence on \(\Sigma\)-forests.
There are three of them:
\begin{itemize}
  \item the class of forests without any $a$,
  \item the class of forests with at least one $a$ whose nodes labelled by $a$ form an antichain,
  \item the class of forests with two comparable nodes labelled by $a$.
\end{itemize}
We can do the same for the equivalence classes of the syntactic equivalence on
\(\Sigma\)-contexts.
This time, there are four of them:
\begin{itemize}
  \item the class of contexts without any $a$,
        \item the class of contexts with at least one $a$, whose nodes labelled
          by $a$ form an antichain and none of these nodes is an ancestor of the \(\square\) node,
        \item the class of contexts with at least one $a$, whose nodes labelled
          by $a$ form an antichain and at least one of them is an ancestor of the \(\square\) node,
  \item the class of contexts with two comparable nodes labelled by $a$.
\end{itemize}

Therefore, the syntactic forest algebra of \(L_{2}\) is \((V_{2},H_{2})\) with \(H_{2}=\left\{\epsilon, a, \treeaa\right\}\) and \(V_{2}=\left\{\square, \treeap, \treea, \treeaas \right\}\).
The tables for each operation are given in~\cref{fig:syntactic_2}.
We omit the table of \(\oplusHV\), as it is again identical to the table of
\(\oplusVH\) whose rows and columns are swapped.
Note that \(H_{2}\) is isomorphic to the monoid \(\{\{0,1,2\},\max\}\).
Let \(v=\treeap\) and \(w=\treea\). Notice that \(v\) is an idempotent and therefore \(v^{\omega+1}=v\).
Thus \(v^{\omega+1}w = \treea \neq \treeaas = wv^{\omega+1}\) and \(L_{2}\) does not belong to ZG.

\begin{figure}
  \centering

  \begin{minipage}{0.45\textwidth}
    \centering
    \[
      \begin{array}{c|ccc}
        \oplusHH & \epsilon & a & \treeaa \\
        \hline
        \epsilon & \epsilon & a & \treeaa \\
        a & a & a & \treeaa\\
        \treeaa & \treeaa & \treeaa&\treeaa
      \end{array}
    \]
  \end{minipage}
  \begin{minipage}{0.45\textwidth}
    \centering
    \[
      \begin{array}{c|cccc}
        \odotVV & \square & \treeap & \treea & \treeaas\\
        \hline
        \square & \square & \treeap&  \treea &  \treeaas  \\
        \treeap & \treeap& \treeap&\treea & \treeaas\\
        \treea &  \treea &\treeaas &\treeaas & \treeaas\\
        \treeaas &  \treeaas & \treeaas& \treeaas& \treeaas\\
      \end{array}
    \]
  \end{minipage}

  \begin{minipage}{0.45\textwidth}
    \centering
    \[
      \begin{array}{c|ccc}
        \odotVH & \epsilon & a & \treeaa \\
        \hline
        \square & \epsilon& a& \treeaa  \\
        \treeap & a & a& \treeaa\\
        \treea & a & \treeaa & \treeaa \\
        \treeaas &  \treeaa & \treeaa& \treeaa\\
      \end{array}
    \]
  \end{minipage}
  \begin{minipage}{0.45\textwidth}
    \centering
    \[
      \begin{array}{c|ccc}
        \oplusVH & \epsilon & a & \treeaa \\
        \hline
        \square & \square & \treeap & \treeaas   \\
        \treeap & \treeap & \treeap & \treeaas \\
        \treea & \treea & \treea  & \treeaas  \\
        \treeaas & \treeaas  & \treeaas& \treeaas\\
      \end{array}
    \]
  \end{minipage}

  \caption{Operation tables of the syntactic forest algebra of $L_2$}
  \label{fig:syntactic_2}
\end{figure}

Finally, consider the \((V_{2},H_{2})\)-forest represented in~\cref{fig:V2_forest}, whose internal nodes are labelled by elements of \(V_{2}\), and leaves are labelled by elements of \(H_{2}\).
To evaluate it, we first consider the function from $H_2$ to $V_2$ that
relabels leaves by mapping \(\epsilon\) to \(\square\) and \(a\) to
\(\treeap\). Applying this function yields a tree where all nodes are  labelled
by elements of \(V_{2}\), which 
is depicted in~\cref{fig:new_V2_forest}.
Let us now apply the
second definition in
the ``Morphisms'' paragraph of~\cref{sec:forest}, and explain how to compute the
image of this forest in $V_2$ by the morphism.
The morphism is defined from a function $g$ from
the nodes of the tree to~$V_2$: here $g$ is simply given by the identity function.
By induction, we evaluate the trees rooted at every node in ascending order:
\begin{itemize}
  \item \textbf{Node 1:} it has the empty sequence of children, so it evaluates to \(\treeap \odot \epsilon = a\).
  \item \textbf{Node 2:} it has the empty sequence of children, so it evaluates to \(\square \odot \epsilon = \epsilon\).
  \item \textbf{Node 3:} it has the empty sequence of children, so it evaluates to \(\square \odot \epsilon = \epsilon\).
  \item \textbf{Node 4:} it has the empty sequence of children, so it evaluates to \(\treeap \odot \epsilon = a\).
  \item \textbf{Node 5:} its forest of children evaluates to \(a\oplus \epsilon=a\), thus it evaluates itself to \(\square \odot a =a \).
  \item \textbf{Node 6:} its forest of children evaluates to \(\epsilon\oplus a=a\), thus it evaluates itself to \(\treeap \odot a = a \).
  \item \textbf{Node 7:} its forest of children evaluates to \(a\oplus a=a\), thus it evaluates itself to \(\treea \odot a = \treeaa \).
\end{itemize}
Note that, in particular, the leaf nodes evaluate back to their label in the original forest
of~\cref{fig:V2_forest}.
Overall, this tree evaluates to \(\treeaa\).

\begin{figure}
  \centering

  \begin{subfigure}{0.45\textwidth}
    \centering
\begin{tikzpicture}[level distance=1.5cm,
  level 1/.style={sibling distance=3cm},
  level 2/.style={sibling distance=1.5cm},
  every node/.style={ draw, minimum size=7mm, inner sep=0pt},
  number/.style={draw = none, font=\tiny,xshift=14pt, yshift=-5pt, text=blue}
]

\node (n1) {\(\treea\) }
  child {node (n2) {\(\square\)}
    child {node[circle] (n4) {\(a\)}}
    child {node[circle] (n5) { \(\epsilon\)}}
  }
  child {node (n3) {\(\treeap\) }
    child {node[circle] (n6) {\(\epsilon\) }}
    child {node[circle] (n7) { \(a\)}}
  };

\node[number] at (n1) {7};
\node[number] at (n2) {5};
\node[number] at (n3) {6};
\node[number] at (n4) {1};
\node[number] at (n5) {2};
\node[number] at (n6) {3};
\node[number] at (n7) {4};

\end{tikzpicture}
    \caption{A \((V_{2},H_{2})\)-tree.}
    \label{fig:V2_forest}
  \end{subfigure}
  \hfill
  \begin{subfigure}{0.45\textwidth}
    \centering
\begin{tikzpicture}[level distance=1.5cm,
  level 1/.style={sibling distance=3cm},
  level 2/.style={sibling distance=1.5cm},
  every node/.style={ draw, minimum size=7mm, inner sep=0pt},
  number/.style={draw = none, font=\tiny,xshift=14pt, yshift=-5pt, text=blue}
]

\node (n1) {\(\treea\) }
  child {node (n2) {\(\square\)}
    child {node (n4) {\(\treeap\)}}
    child {node (n5) { \(\square\)}}
  }
  child {node (n3) {\(\treeap\) }
    child {node (n6) {\(\square\) }}
    child {node (n7) { \(\treeap\)}}
  };

\node[number] at (n1) {7};
\node[number] at (n2) {5};
\node[number] at (n3) {6};
\node[number] at (n4) {1};
\node[number] at (n5) {2};
\node[number] at (n6) {3};
\node[number] at (n7) {4};

\end{tikzpicture}
    \caption{Replacement of the leaves.}
    \label{fig:new_V2_forest}
  \end{subfigure}

  \caption{Evaluation of a \((V_{2},H_{2})\)-tree.}
\end{figure}

\end{toappendix}

\section{Forest Algebras}
\label{sec:forest}
All our results about the dynamic membership problem are in fact shown
by rephrasing to the terminology of \emph{forest algebras}. Intuitively, forest
algebras give an analogue for trees to the monoids and syntactic morphisms used
in the setting of words, which formed the first step of the classification of
the complexity of dynamic membership on
words~\cite{frandsen1995dynamic,amarilli2021dynamic}.
More details on forest algebras are given in Appendix~\ref{sec:algebra}, with in particular worked-out examples of computations of forest algebras and of the evaluation of a tree labelled by a forest algebra.

\subparagraph*{Monoids.}
A \emph{monoid} $M$ is a set equipped with an associative composition law
featuring a \emph{neutral element}, that is, an element $e$ such that $ex = xe =
x$ for all $x$ in~$M$. We write the composition law of monoids multiplicatively,
i.e., we write $xy$ for the composition of~$x$ and~$y$.
The \emph{idempotent power} of an element $v$ in a finite monoid $M$, written
$v^\omega$, is $v$ raised to the least
integer $\omega$ such that $v^\omega = v^\omega v^\omega$: see
\cite[Chp.\ II, Prop.\ 6.31]{mpri} for a proof that every element has an
idempotent power. The \emph{idempotent power} of~$M$ is a value $m \in \mathbb{N}$
ensuring that, for each
$v \in M$, we have $v^m v^m = v^m$: this can be achieved by taking any common
multiple of the exponents $\omega$ for the idempotent powers of all elements
of~$M$.

\subparagraph*{Forest algebra.}
A forest algebra~\cite{bojanczyk2008forest} is a pair $(V, H)$ of two monoids.
The monoid $V$ is the \emph{vertical monoid}, with vertical composition written
$\odotVV$ and neutral element written $\square$.
The monoid $H$ is the \emph{horizontal monoid}, 
with horizontal concatenation denoted $\oplusHH$,
and with neutral element written $\epsilon$. 
Further, we have three composition laws:
    $\odotVH \colon V \times H \to H$, 
    $\oplusVH\colon V \times H \to V$, and
    $\oplusHV \colon H \times V \to V$.
We require that the following relations hold:
\begin{itemize}
  \item \emph{Action (composition)}: for every \(v,w\in V\) and \(h\in H\),
    \((v\odotVV w)\odotVH h = v\odotVH (w\odotVH h)\).
  \item \emph{Action (neutral)}: for every \(h\in H\), \(\square \odotVH h = h\).
  \item \emph{Mixing}: for every \(v\in V\) and \(h,g\in H\), 
    we have \((v\oplusVH h)\odotVH g = (v\odotVH g)\oplusHH h\) and
    \((h\oplusHV v)\odotVH g = h \oplusHH (v\odotVH g)\).
  \item \emph{Faithfulness}: for every distinct \(v,w\in V\), there exists
    \(h\in H\) such that \(v\odotVH  h \neq w\odotVH h\).
\end{itemize}
We explain in Appendix~\ref{sec:algebra} how this formalism slightly differs
from that of~\cite{bojanczyk2008forest} but is in fact equivalent.
We often abuse notation and write $\oplus$ to mean one of $\oplusHH,
\oplusHV, \oplusVH$ and write $\odot$ to mean one of $\odotVH,
\odotVV$.

We almost always assume that forest algebras are \emph{finite}, i.e., 
the horizontal and vertical monoids $H$ and $V$ are both finite. The only exception is
the \emph{free forest algebra} \(\Sigma^{\nabla}= (\SigmaV, \SigmaH)\).
Here,
$\SigmaH$ is the set of all $\Sigma$-forests, and 
$\SigmaV$ is the set of all \emph{$\Sigma$-contexts}, i.e., $\Sigma$-forests
having precisely one node that
carries the special label $\square \notin \Sigma$: further, this node must be a leaf.
The \(\oplus\)-laws denote horizontal \emph{concatenation}: of two forests
for $\oplusHH$, and of a forest and a context for $\oplusHV$ and $\oplusVH$.
(Note that two contexts cannot be horizontally concatenated because the result
would have two nodes labeled~$\square$.) We write the
$\oplus$-laws as $+$, and remark that they are not commutative. 
The \(\odot\)-laws are the \emph{context application operations} of plugging the forest (for $\odotVH$) or
context (for $\odotVV$) on the right
in place of the \(\square\)-node of the context on the left. We write the
$\odot$-laws functionally, i.e., for $s \in \SigmaV$ and for 
$u \in
\SigmaV$ or $f \in \SigmaH$ we write $s(u)$ to mean $s \odotVV u$ and
write $s(f)$ to mean $s \odotVH f$.

We write $\square \in \SigmaV$ for the trivial
$\Sigma$-context consisting only of a single node labeled~$\square$, we write
$\epsilon \in \SigmaH$ for the empty forest, and for $a
\in \Sigma$ we write $a_\square \in \SigmaV$ the $\Sigma$-context which consists of a single
root node labeled $a$ with a single child labeled~$\square$. Note that
$\SigmaV$ and $\SigmaH$ are spanned by $\epsilon$ and $\square$ and by
the $a_\square$ via context application and concatenation.

\subparagraph*{Morphisms.}
A \emph{morphism of forest algebras}~\cite{bojanczyk2008forest} consists of two functions that map forests to
forests and contexts to contexts and are compatible with the internal operations.
Formally, a morphism from \((V,H)\) to \((V',H')\) is a pair of functions
\(\muV\colon V \to V',\muH\colon H\to H'\) where:
\begin{itemize}
  \item \(\muV\) is a monoid morphism: for all \(v,w\in V\), we have \(\muV(v\cdot w) =
    \muV(v)\cdot \muV(w)\) and \(\muV(\square)=\square\) where \(\cdot\) and \(\square\)
        are interpreted in the corresponding monoid.
  \item \(\muH\) is a monoid morphism: for all \(h,g\in H\), we have \(\muH(h+ g) =
    \muH(h)+ \muH(g)\) and \(\muH(\epsilon)=\epsilon\) where \(+\) and \(\epsilon\)
        are interpreted in the corresponding monoid.
  \item for every \(v\in V\), \(h\in H\), we have \(\muH(v\odotVH h) = \muV(v)
    \odotVH  \muH(h)\).
  \item for every \(v\in V\), \(h\in H\), we have \(\muV(v\oplusVH h) =
    \muV(v) \oplusVH \muH(h)\) and \(\muV(h\oplusHV v) = \muH(h)
    \oplusHV \muV(v)\).
\end{itemize}
We often abuse notation and write a morphism $\mu$ to mean the function
that applies $\muV$ to $\Sigma$-contexts and $\muH$ to $\Sigma$-forests.
For any alphabet $\Sigma$ and forest algebra $(V, H)$,
given a function $g$ from the alphabet $\Sigma$ to~$V$, we extend it to 
a morphism $\muV,\muH$ from the free forest algebra to $(V,H)$
in the only way compatible with the requirements above, i.e.:
\begin{itemize}
\item For a sequence of $\Sigma$-trees $F = t_1, \ldots, t_k$ where at most one $t_i$
  is a $\Sigma$-context, letting $x_i \coloneq \mu(t_i)$ for each $i$ be inductively
    defined, then
    $\mu(F) \coloneq x_1 \oplus \cdots \oplus x_k$.
\item For a $\Sigma$-tree or $\Sigma$-context $t$ with root node $u$ labeled $a \in \Sigma$,
  letting $F$ be the (possibly empty) sequence of children of~$u$
    (of which at most one is a context, and for which 
    $\mu(F)$ was inductively
      defined),
      then we set $\mu(t) \coloneq (g(a)) \odot \mu(F)$.
\end{itemize}
A forest language $L$ over $\Sigma$ is \emph{recognized} by a forest algebra $(V,
H)$ if there is a subset $H' \subseteq H$ and a function from $\Sigma$ to~$V$
defining a morphism $\muV,\muH$ having the following property: a $\Sigma$-forest $F$
is in $L$ iff $\muH(F) \in H'$.

\subparagraph*{Syntactic forest algebra.}
Regular forest languages can be related to forest algebras via the notion of
\emph{syntactic forest algebra}.
Indeed, 
a forest language is
regular iff it is recognized by some forest algebra (see 
\cite[Proposition 3.19]{bojanczyk2008forest}). Specifically, 
we will consider the \emph{syntactic forest algebra} of a regular forest
language~$L$: this forest algebra recognizes $L$, it is minimal in a certain
sense, and it is unique up to isomorphism. We omit the formal definition of the
syntactic forest algebra $(V,H)$ of~$L$ (see \cite[Definition~3.13]{bojanczyk2008forest} for details).
We will just use the fact that it recognizes~$L$ for a certain function $g$ from~$\Sigma$ to~$V$ and
    associated morphism~$\muV,\muH$ from $\Sigma$-forests to
    $(V,H)$, called the \emph{syntactic morphism}, and satisfying the following:
\begin{itemize}
  \item \emph{Surjectivity:} for any element $v$ of~$V$, there is a
    $\Sigma$-context $c$ such that $\muV(c) = v$;
  \item \emph{Minimality:} for any two $\Sigma$-contexts $c$ and $c'$, if $\muV(c) \neq
    \muV(c')$, then there is a $\Sigma$-context $r$ and a $\Sigma$-forest $s$
    such that exactly one of $r(c(s))$ and $r(c'(s))$ belongs to~$L$.
\end{itemize}

\subparagraph*{Dynamic evaluation problem for forest algebras.}
We will study the analogue of the dynamic membership problem for forest
algebras, which is that of computing the \emph{evaluation}
of an expression. More precisely, for $(V,H)$ a
forest algebra, a \emph{$(V,H)$-forest} is a forest 
where each internal node is labeled by an element of~$V$
and where each leaf is labeled with an element of~$H$ -- but there may be one
leaf, called the \emph{distinguished leaf}, which is labeled with an element of~$V$.
The \emph{evaluation} of a $(V,H)$-forest $F$ is the image of $F$
by the morphism obtained by extending the function 
$g$ which maps elements of~$V$ to themselves and maps elements $f$ of~$H$ to the
context $c_f \coloneq f \oplusHV \square$ (so that $c_f \odotVH \epsilon = f$).
Remark that the forest $F$ evaluates to an element of~$V$ if it has a distinguished leaf,
and to~$H$ otherwise.

The \emph{(non-restricted) dynamic evaluation problem} for $(V,H)$ then
asks us to maintain the evaluation of the $(V,H)$-forest $F$ under
relabeling updates which can change the label of internal nodes of~$F$
(in~$V$) and of leaf nodes of~$F$ (in~$H$ or in~$V$, but ensuring that there
is always at most one distinguished leaf in~$F$).  Again, we assume that the
forest algebra $(V,H)$ is constant, and we measure the complexity 
as a function of the size $|F|$ of the input forest 
(note that updates never change~$|F|$ or the shape of~$F$).
The \emph{restricted dynamic evaluation problem} for $(V,H)$
adds the restriction that the label of leaves always stays in $H$ (initially and
after all updates), so that $F$ always evaluates to an element of~$H$.

We will use the dynamic evaluation problem in the next section for our
$O(\log n / \log \log n)$ upper bound. Indeed, it generalizes the dynamic membership
problem in the following sense:

\begin{lemmarep}
  \label{lem:regular2forest}
  Let $L$ be a fixed regular forest language, and let $(V,H)$ be its syntactic
  forest algebra. 
  Given an algorithm for the restricted dynamic evaluation problem for~$(V,H)$ under
  relabeling updates, we can obtain an algorithm for the dynamic membership
  problem for~$L$ with the same complexity per update.
\end{lemmarep}

\begin{proof}
  As the syntactic forest algebra recognizes $L$,
  let $\mu$ be the morphism from $\Sigma$-forests and $\Sigma$-contexts to
  $(V,H)$, and let $H'$ be the subset of $H$ to which $\Sigma$-forests in~$L$
  are mapped by~$\mu$.
  Let $F$ be the input $\Sigma$-forest. In the preprocessing, we prepare a
  $(V,H)$-forest $F'$ by
  translating annotations: every leaf of~$F$ labeled with $a \in \Sigma$ is
  labeled in~$F'$ with the element $\muH(a)$ of~$H$, and every internal node
  of~$F$ labeled with $a \in \Sigma$ is labeled in~$F'$ with the image by~$\muV$
  of the $\Sigma$-context $a_\square$. Every update on~$F$ is translated to an
  update on~$F'$ in the expected way. Now, from the evaluation of the
  $(V,H)$-forest $F'$ maintained by the dynamic evaluation problem data
  structure on~$F'$, we obtain an element $h \in H$. Looking at the definitions
  of the morphism $\mu$ and of the evaluation of a $(V,H)$-forest, an
  immediate induction shows that $h$ is the image of~$F$ by~$\mu$.
  Thus, testing in constant time whether $h \in H'$
  allows us to deduce whether $F$ is in~$L$.
\end{proof}

\section{Dynamic Membership to Regular Languages in $O(\log n / \log \log n)$}
\label{sec:lll}
Having defined the technical prerequisites, we now
start the presentation of our technical results. In this section, we show our
general upper bound on the dynamic membership problem to arbitrary regular
forest languages:

\begin{theorem}
  \label{thm:upper}
  For any fixed regular forest language $L$, the dynamic membership problem
  to~$L$ is in $O(\log n / \log \log n)$, where $n$ is the number of nodes of
  the input forest.
\end{theorem}

Note that this upper bound matches the lower bounds reviewed at the end of
\cref{sec:prelim}. We present the proof in the rest of this
section. By \cref{lem:regular2forest}, we will instead show
our upper bound on the restricted dynamic evaluation problem for arbitrary fixed forest algebras.

The algorithm that shows \cref{thm:upper} intuitively follows a recursive
scheme. 
For the first step of the scheme, given the input forest $\cF_0$, we compute a
so-called \emph{clustering} of~$\cF_0$. This is a partition of the nodes of~$\cF_0$ into
subsets, called \emph{clusters}, which are connected in a certain sense and will
be chosen to have size $O(\log n)$.
Intuitively, clusters are small enough so
that we maintain their evaluation under updates in $O(1)$ by tabulation;
note that clusters may correspond to contexts (i.e., they may have holes), so we
will perform \emph{non-restricted} dynamic evaluation for them.
Further, $\cF_0$ induces a forest
structure on the clusters, called the \emph{forest of clusters} and denoted
$\cF_1$, for which we will ensure that it has size $O(n / \log n)$.
We note that the clustering algorithm that we propose is very similar to the one constructed in the
work of Gazit et al.~\cite{gazit1988optimal}, which has similar properties
and could be used as an
alternative to the one we propose. However, our clustering is not constructed in
the same way and the result of building it on an input tree is not identical
to the clustering of~\cite{gazit1988optimal}, so we present our own construction to
ensure that our work is self-contained.

Having computed the clustering of~$\cF_0$ and the forest of clusters $\cF_1$, we
re-apply recursively the clustering scheme on $\cF_1$, decomposing it again into
clusters of size $O(\log n)$ and a forest of clusters $\cF_2$ of size $O(n / (\log
n)^2)$. We recurse until we obtain a forest $\cF_\ell$ with only one node, which
is the base case: we will ensure that $\ell$ is in $O(\log n / \log \log n)$.

To handle updates on a node $u$ of~$\cF_0$, we will
apply the update on the cluster $C$ of~$\cF_0$ containing~$u$ (in $O(1)$ by
tabulation), and apply the resulting update on the node $C$ in the forest of
clusters $\cF_1$. We then continue this process recursively, eventually reaching
the singleton $\cF_\ell$ where the update is trivial.
The main technical
challenge is to bound the complexity of the preprocessing: we must show how to
efficiently compute a suitable clustering on an
input forest $\cF$ in time $O(|\cF|)$. It will then be possible to apply the
algorithm to $\cF_0, \cF_1, \ldots, \cF_{\ell-1}$, with a total complexity amounting to
$O(|\cF_0|)$.

The section is structured as follows. First, we formally define the notion of
\emph{clusters} and \emph{clustering} of a forest~$\cF$, and we
define the \emph{forest of clusters} induced by a clustering: note that these
notions only depend on the shape of~$\cF$ and not on the node labels.
Second, we explain
how the evaluation of a $(V,H)$-forest reduces to computing the
\emph{evaluation} of clusters along with the \emph{evaluation} of the forest of
clusters. Third, we explain how we can compute in linear time a clustering of the 
input forest which ensures that the forest of clusters is small enough: we show
that it is sufficient to compute any \emph{saturated} clustering (i.e., no clusters can
be \emph{merged}), and sketch an algorithm that achieves this.
Fourth, we conclude the
proof of \cref{thm:upper} by summarizing how the recursive scheme works,
including how the linear-time preprocessing can tabulate the effect of updates on small
forests.

\subparagraph*{Clusters and clusterings.}
A \emph{clustering} of a forest will be defined by partitioning its vertices
into
\emph{connected} sets, where connectedness is defined using the
sibling and first-child edges.

\begin{definition}
  Let $\cF$ be a forest. We say that two nodes of~$\cF$ are
  \emph{LCRS-adjacent}
  (for \emph{left-child-right-sibling})
  if one node is the first child of the other or
  if they are two consecutive siblings (in particular if they are two
  consecutive roots).
  We say that a set of nodes
  of~$\cF$ is \emph{LCRS-connected} if,
  for any two nodes $u, u'$
  in~$\cF$, there is a sequence $u = u_1, \ldots, u_q = u'$  of
  nodes in~$\cF$ such that $u_i$ and $u_{i+1}$ are LCRS-adjacent
  for each $1 \leq i < q$.
\end{definition}

Note that the edges used in LCRS-adjacency are \emph{not} the
edges of~$\cF$, but those of a left-child-right-sibling representation
of~$\cF$ (hence the name). 
For instance, the set $\{u, u'\} \subseteq \cF$ formed of a node $u$ and its
parent $u'$ is \emph{not} connected unless $u$ is the first child of~$u'$.
To define clusters and clusterings, we will use LCRS-adjacency, together with a
notion of \emph{border nodes} that correspond to the ``holes'' of clusters:

\begin{definition}
  \label{def:kred}
Given a forest $\cF$ with $n$ nodes, we say that a node $u$
  in a subset $S$ of $\cF$ is a \emph{border node} of~$S$ if $u$ has a child
  which is not in $S$.
  For $k > 0$, we then say that an equivalence relation $\equiv$ over
  the nodes of $\cF$ is a \emph{$k$-clustering} when the following
  properties hold on the equivalence classes of~$\equiv$, called
  \emph{clusters}:
  \begin{itemize}
    \item each cluster contains at most $k$ nodes;
    \item each cluster is LCRS-connected;
    \item each cluster contains at most one border node.
  \end{itemize}
  The \emph{roots} of~$S$ are the nodes of~$S$ whose parent is
  not in~$S$ (or which are roots of~$\cF$): 
  if $S$ is a cluster, then by LCRS-connectedness its roots must be consecutive
  siblings in~$\cF$.
\end{definition}

When we have defined a $k$-clustering, it induces a \emph{forest of clusters} in
the following way:

\begin{definition}
  Given a $k$-clustering $\equiv$ of a forest $\cF$, the
  \emph{forest of clusters} $\cF^{\equiv}$ is a forest whose nodes 
  are the clusters of~$\equiv$, and where a cluster $C_1$ is the child of a
  cluster $C_2$ when the roots of $C_1$ are children of the border
  node of $C_2$.

  We order the children of a cluster $C$ in~$\cF^\equiv$ in the
  following way. For each child $C'$ of~$C$, its root nodes are a set of
  consecutive siblings, and these roots are in fact consecutive children of the
  border node~$u$ of~$C$. Thus, given two children $C_1$ and $C_2$ of~$C$ in
  $\cF^\equiv$, we order $C_1$ before $C_2$ if the roots of $C_1$ come
  before the roots of $C_2$ in the order in~$\cF$ on the children of~$u$.
  Likewise, we can order the roots of~$\cF^\equiv$, called the
\emph{root clusters}, according to the order on the roots
of~$\cF$: recall that, by our definition of siblings, the root clusters
  are also siblings in~$\cF^\equiv$.
\end{definition}

Remark that the trivial equivalence relation (putting each
node in its own cluster) is vacuously a $k$-clustering in which
the border nodes are precisely the internal nodes: we call this the
\emph{trivial $k$-clustering}, and its forest of clusters is isomorphic
to~$\cF$.

\subparagraph*{Evaluation of clusters.}
\begin{toappendix}
  \subsection{Evaluation of Clusters}
\end{toappendix}
To solve the dynamic evaluation problem on~$\cF$ using a
clustering~$\equiv$, we will now explain how the evaluation of the
$(V,H)$-forest $\cF$ can reduce to the evaluation of the forest of
clusters $\cF^\equiv$ with a suitable labeling. To define this labeling,
let us first define the \emph{evaluation} of a cluster in~$\cF$:

\begin{definition}
  \label{def:clusteval}
Given a $(V,H)$-forest $\cF$ with no distinguished leaf,
  a $k$-clustering $\equiv$ of $\cF$, and a cluster $C$ of~$\equiv$, we
  define the \emph{evaluation of~$C$} as a value in $V$ or~$H$ in the following
  way. Let $\cF^C$ be the
sub-forest of $\cF$ induced by~$C$, i.e., the sub-forest containing only
  the nodes in~$C$ and the edges connecting two nodes that are both in~$C$: note
  that it is a $(V,H)$-forest where each node has the same label as
  in~$\cF$.
  When $C$ contains a border node $u$, we
also add a leaf $u'$ as the last child of~$u$
  in $\cF^C$ and label $u'$ with $\square \in V$.

  The \emph{evaluation} of the cluster $C$ in~$\cF$ is then
  the evaluation of $\cF^C$ as a $(V,H)$-forest. Note that, as $\cF$ has
  no distinguished leaf, the evaluation is in~$V$ if $C$ has a border
  node and in~$H$ otherwise; in other words it is in $V$ exactly when
  $C$ has a child in $\cF^\equiv$.
\end{definition}

We now see 
$\cF^\equiv$ as a $(V,H)$-forest,
with each cluster labeled by its evaluation in~$\cF$:
Then:

\begin{claimrep}
  \label{clm:correct}
  For any $k$-clustering $\equiv$ of a $(V,H)$-forest $\cF$,
  the evaluation of~$\cF$
  is the same as the evaluation of its forest of clusters $\cF^\equiv$.
\end{claimrep}

\begin{proof}
We will prove this claim inductively on the size of
$\cF^\equiv$. When there is a single cluster in $\cF$,
the definition of the evaluation of a cluster gives us directly the
result.

Let us now consider the case where $\cF^\equiv$ has multiple
roots $c_1 \cdots c_q $, then $\cF$ must also have multiple
roots $r_1 \dots r_\ell$ ($\ell\geq q$). The evaluation of
$\cF^\equiv$ is computed as $h^c_1 \oplus \cdots \oplus h^c_q$
where each $h^c_i$ is the evaluation of the tree rooted in
$c_i$. Similarly, the evaluation of $\cF$ is computed as
$h^r_1 \oplus \cdots \oplus h^r_\ell$ where the $h^r_i$ is the
evaluation of the tree rooted in $r_i$. Using the associativity of
$\oplus$ and the fact that clusters are connected, we can rewrite this
term as $h^g_1 \oplus \cdots \oplus h^g_q$ where $h^g_i$ is the sum
$h^r_{s_i} \oplus \cdots \oplus h^g_{e_i}$ when $c_i$ is the cluster
containing the roots $r_{s_i}$ to $r_{e_i}$. For each
$1\leq i \leq q$
we can apply the inductive property on each subtree of
$\cF$ that contain the roots in $r_{s_i}, \dots, r_{e_i}$ plus
their descendant and the clustering $\equiv$ restricted to this
sub-forest and we obtain the result that the evaluation of this
sub-forest is $h^g_i$.

Finally let us consider the case where $\cF^\equiv$ has a
single root $C_0$ with children $C_1 \dots C_q$.
  By definition of the
$k$-clustering, the roots of $C_1 \dots C_q$ must all be children of
the same node $u$ in $C$, the border node of $C$. This node $u$ has
the children $u_1 \cdots u_\ell$ and they are in clusters $C_0, \dots,
C_q$. For each $1 \leq i \leq \ell$, let us write $D_i$ for the
  contiguous subsequence of the $u_1 \dots u_\ell$
that are in cluster $C_i$. The evaluation of $u$ in
$\cF$ is computed as $v \odot (h^{u_1} \oplus \cdots \oplus
h^{u_\ell} )$ where $v$ is the label of $u$ and $h^{u_i}$ is the
evaluation of the tree rooted in~$u_i$. Using the associativity of $\oplus$ we can
group the $u_i$ that are in the same cluster and rewrite the
evaluation of $u$ into $v \odot (h^D_0 \oplus \cdots \oplus h^D_q) =
(v \odot (h^D_0 \oplus \square )) \odot (h^D_1 \oplus \cdots \oplus
h^D_q)$ where each $h^D_i$ is the sum of the $h^{u_j}$ for $u_j$ in $D_i$: we
use the induction hypothesis to show that $h^D_i$ for $i\geq 1$ is
equal to the evaluation of cluster $C_i$. Now using the fact that the evaluation is a morphism,
we have that the evaluation $h^C_0$ of~$C_0$ and the evaluation
$h^\cF$ of $\cF$ are such that $h^\cF = h^C_0
\odot (h^D_1 \oplus \cdots h^D_q)$ which is also the evaluation of
$\cF^\equiv$.
\end{proof}

This property is what we use in the recursive scheme: to solve the dynamic
evaluation problem on the input $(V,H)$-forest $\cF$, during the
preprocessing we will compute a clustering $\equiv$ of~$\cF$ and compute
the evaluation of clusters and of the forest of clusters $\cF^\equiv$.
Given a relabeling update on~$\cF$, we will apply it to its cluster $C$ and
recompute the evaluation of~$C$; this gives us a relabeling update to apply to
the node $C$ on the forest of clusters $\cF^\equiv$, and we can
recursively apply the scheme to~$\cF^\equiv$ to maintain the evaluation
of~$\cF^\equiv$.
What remains is to explain how we can efficiently compute a clustering
of a forest~$\cF$ that ensures that the forest of clusters $\cF^\equiv$ is ``small
enough''.

\subparagraph*{Efficient computation of a clustering.}
\begin{toappendix}
  \subsection{Saturated Clusterings}
  We will start with some observations on mergeable clusters and saturated
  clusterings, which will be useful for the rest of the proof, and will be
  used later in the presentation of the algorithm:

 \begin{remark}
  \label{rem:mergeable}
  Two distinct clusters are mergeable exactly when:
  \begin{itemize}
    \item they are LCRS-adjacent in the forest of clusters
    \item their total size is less than $k$
    \item the resulting merge contains at most one border node.
  \end{itemize}
\end{remark}

The first point of the above remark is equivalent to requiring that the
resulting cluster is LCRS-connected, which hinges on the following claims:

\begin{claim}
  \label{claim:connectedness}
  For any two distinct clusters $C$ and $C'$, the following are
  equivalent:
  \begin{itemize}
    \item $C \cup C'$ is LCRS-connected in~$\cT$;
    \item there is a node $u$ of~$C$ and a node $u'$ of~$C'$ which are LCRS-adjacent
  in~$\cT$.
  \end{itemize}
\end{claim}

\begin{proof}
If we suppose that $C\cup C'$ is LCRS-connected, we can take $u\in C$,
$v\in C'$ and find a sequence $z_1, \dots, z_q$ with $z_1=u$, $z_q=v$
in $C'$ such that $z_i \in C\cup C'$ for all $1 \leq i \leq q$ and such that
  $z_i$ and $z_{i+1}$
  LCRS-adjacent for each $1 \leq i < q$.
  Note that $q > 1$ because $C$ and $C'$ are distinct, hence disjoint.
  By taking the first $i\geq 1$ such that $z_{i+1}\in C'$ (which
  is well-defined because the sequence finishes in~$C'$)
we find the node $u=z_i$ of~$C$ and the node $u'=z_{i+1}$ of~$C'$
which are LCRS-adjacent in~$\cT$.

Conversely if we have a node $u$ of~$C$ and a node $u'$ of~$C'$ which
are LCRS-adjacent then to prove the LCRS-connectedness of $C\cup C'$
it suffices to prove that for all $(v,v')\in (C\cup C')^2$ we have a
sequence $v = z_1, \dots, z_q=v'$ such that the $z_i, z_{i+1}$ are
LCRS-adjacent for each $1\leq i < q$. If $v$ and $v'$ both fall in $C$
or $C'$, we have the result by the LCRS-connectedness of $C$ and
$C'$. Without loss of generality let us consider the case $v\in C$,
$v'\in C'$. Here, we can construct the sequence as the sequence from
$v$ to $u$ concatenated with the sequence from $u'$ to $v'$.
\end{proof}

\begin{claim}
  \label{clm:equiv2}
  For any two distinct clusters $C$ and $C'$, the following are equivalent:
  \begin{itemize}
    \item $C \cup C'$ is LCRS-connected in~$\cT$
    \item $C$ and $C'$ are LCRS-adjacent in the forest of clusters
  \end{itemize}
\end{claim}

\begin{proof}
Let us suppose that $C\cup C'$ is LCRS-connected in~$\cT$. By
\cref{claim:connectedness} we have a node $u$ of~$C$ and a node
$u'$ of~$C'$ which are LCRS-adjacent in~$\cT$. Up to
exchanging $(u,C)$ and $(u',C')$, we can suppose that $u'$ is the
child or next-sibling of $u$. If $u'$ is the child of $u$, then $u$ is
the border node and $C'$ is a child of $C$. If $u'$ is the next-sibling
of $u$ then either $C$ contains the shared parent of $u$ and $u'$ and
$C'$ is a child of $C$ or it does not (in particular when it does not exist
because $u$ and $u'$ are roots) and $C'$ is the next-sibling of
$C$. In all cases, $C$ and $C'$ are LCRS-adjacent.

Conversely, let us suppose that $C$ and $C'$ are LCRS-adjacent in the
forest of clusters. Without loss of generality, let us consider that
$C'$ is the child or next-sibling of $C$. When $C'$ is the child of
$C$, it means that the first root $r$ of $C'$ is the leftmost child of
the border node $u$ of $C$ that is not in $C$ (e.g., $r$ can be the first
child of $u$ or the second child when the first is in $c$).  When $C'$
is the next-sibling of $C$ then the first root of $C'$ is the
next-sibling of the last root of $C$. In both case, we have two nodes
in $C$ and $C'$ which are LCRS-adjacent and therefore by
\cref{claim:connectedness} we have that $C\cup C'$ is
LCRS-connected in $\cT$.
\end{proof}

The two previous claims deal with the LCRS-adjacency conditions. For clusters
to be mergeable, there are two other conditions: the total size and
the number of border nodes. Checking the size is easy (it is the sum)
but in the presentation of the algorithm we will need the following immediate claim to take care of the number of
border nodes:

\begin{claim}
  \label{clm:mergeborder}
  Let $C_1$ and $C_2$ be two clusters and let $B_1$ and $B_2$ be their
  respective border nodes (with $|B_1| \leq 1$ and $|B_2| \leq 1$). Then the
  border nodes of~$C_1 \cup C_2$ are precisely $B_1 \cup B_2$, 
  except in the following situation (or its
  symmetric up to exchanging $C_1$ and $C_2$): $C_1$ has a border node $u$ and
  all children of~$u$ are in $C_1 \cup C_2$. Equivalently, all roots of~$C_2$
  are children of~$u$, all preceding siblings of these roots are in~$C_1$, and
  there are no following siblings of these roots.
\end{claim}

  We can now re-state and prove \cref{prp:saturated}:
\end{toappendix}

Here is what we want to establish:

\begin{proposition}
  \label{prop:denseClustering}
  There is a fixed constant $c \in \NN$ such that the following is true:
  given a forest $\cF$, we can
  compute a $k$-clustering $\equiv$ of~$\cF$ in linear time 
  such that
  $|\cF^\equiv| \leq  \lceil |\cF| \times c/k \rceil $.
\end{proposition}

Our algorithm starts with the trivial $k$-clustering and iteratively merges
clusters. 
Formally, \emph{merging} two
clusters $C$ and $C'$ means setting them to be equivalent
under~$\equiv$, and we call $C$ and $C'$
\emph{mergeable} when doing so leads to a $k$-clustering. Of course, we will
only merge mergeable clusters in the algorithm.

We say that the $k$-clustering is \emph{saturated} if it does not
contain two mergeable clusters.  Our definition of
clusters is designed to ensure that, as soon as the clustering is
saturated, there are enough large clusters so that the number of
clusters (i.e., the size of the forest of clusters) satisfies the
bound of \cref{prop:denseClustering}, no matter how clusters were
merged. Namely:

\begin{claimrep}
  \label{prp:saturated}
  There is a fixed constant $c \in \NN$ such that the following is true:
  given a
$(V,H)$-forest $\cF$, any saturated
$k$-clustering on $\cF$ has at most $\lceil c \times (n/k) \rceil$
clusters.
\end{claimrep}

\begin{proofsketch}
  We focus on clusters with zero or one child in the forest of
  clusters.  We consider potential merges of these clusters with their
  only child, with their preceding sibling, or with their parent (if
  they are the first child).  This allows us to show that, provided
  that there are multiple clusters, a constant fraction of them have
  to contain at least $k/2$ nodes, so the number of clusters is
$O(n/k)$.
\end{proofsketch}

\begin{proof}
Our proof will proceed by a case analysis according to the number
of children of each cluster in the forest of clusters. It is easy to see that,
  in any tree with $k$ leaves, there are at most $k-1$ nodes with at least two
  children. For this reason, let us call $N_0$ the number of clusters
  with no child, $N_1$ the number of clusters with a single child, 
  $N_{\geq 2}$ the number of clusters with 2 or more children, and $N$
  the total number of clusters. We have $N = N_0 + N_1 + N_{\geq 2} \leq 2 N_0 +
  N_1 \leq 2 (N_0 + N_1)$. So, to show the claimed upper bound on~$N$, we will
  show an upper bound on $N_0+N_1$ in the sequel.

For a cluster $C$ with a single child $C'$, since the $k$-clustering is
saturated, we know that $C$ cannot be merged with~$C'$. 
  By \cref{rem:mergeable}, as $C'$ is the only child of~$C$, if they are not
  mergeable, it must be because $|C \cup C'| > k$. 
  Let us show that, for this reason, we have $N_1 \leq 2n/k$.
  Let us sum, over all clusters with a single child, the cardinality
  of the cluster unioned with their single child. Each term of the sum is greater
  than~$k$, and there are $N_1$ terms in the sum, so the sum is 
  $\geq N_1 k$. But, in the sum, each node of~$\cF$ is summed at most twice:
  because clusters form a partition of~$\cF$, a node can only be
  summed in the cluster that contains it, which occurs in at most two terms.
  So the sum is $\leq 2n$. This implies that $N_1 k \leq 2n$, thus
  $N_1 \leq 2n/k$.

For a cluster $C$ with 0 children, there are three mutually exclusive cases:

  \begin{itemize}
  \item $C$ is the first child of a cluster;
  \item $C$ has a preceding sibling;
  \item Neither of these cases apply.
  \end{itemize}

  We exclude the third case, as it concerns only a single cluster at most,
  namely, the first root of the forest of clusters (and only if it has 0
  children). So, let us bound
  the number of clusters of each of the two first cases: write $S_{01}$ and
  $S_{02}$ for the corresponding sets of clusters, with
  $N_0 \leq |S_{01}| + |S_{02}| + 1$.

  If a cluster $C$ is the first child of $C'$, they are LCRS-adjacent and
  $C$ has no border node, so by \cref{rem:mergeable} if they cannot be merged
  then $|C \cup C'| > k$. Like in the proof for~$N_1$, let us sum, for $C \in
  S_{01}$, the cardinalities of the union of $C$ with its parent: we conclude that $|S_{01}| \leq 2n/k$.

  If a cluster $C$ is the next sibling of~$C'$ in the forest of clusters, then
  $C$ and $C'$ are LCRS-adjacent in the forest of clusters, and $C$ has no
  border node. So, again by \cref{rem:mergeable}, we have $|C \cup C'| > k$,
  and summing for $C \in S_{02}$ the cardinalities of the union of~$C$ with
  its preceding sibling, we conclude $|S_{02}| \leq 2n/k$.

  Overall, we have shown that the number of clusters is $O(n/k)$, concluding
  the proof.
\end{proof}

\begin{toappendix}
\subsection{Computing a Saturated $k$-Clustering}
\label{sec:compSaturatedClustering}
We can now give the detailed algorithm to show the following claim, along with
  its correctness proof:
\end{toappendix}

Thus,
to show \cref{prop:denseClustering}, it suffices to compute a saturated
\mbox{$k$-clustering}. Namely:

\begin{claimrep}
  \label{clm:algo}
Given a forest, we can compute a saturated $k$-clustering
along with its forest of clusters in linear time.
\end{claimrep}

\begin{proofsketch}
We start with the trivial $k$-clustering, and then we saturate the clustering by
exploring nodes following a standard depth-first traversal: from a node $u$, we proceed recursively on
  the first child $u'$ of~$u$ (if it exists) and then on the next sibling $u''$ of~$u$ (if it
  exists). Then, we first try to
  merge the cluster of $u$ with the cluster of the first child $u'$, and then
  try to merge the cluster of $u$ with the cluster of the next
  sibling $u''$. By ``try to merge'', we mean
that we merge the two clusters if they are mergeable.

This algorithm clearly runs in linear time as long as the mergeability tests and the
actual merges are performed in constant time, which we can ensure with the right
data structures. The algorithm clearly returns a clustering, and the complicated
  part is to show that it is saturated. For this, assume by contradiction that
  the result contains two mergeable clusters with LCRS-adjacent nodes $u_1$ and
  $u_2$: without loss of generality $u_2$ is the first child or next sibling
  of~$u_1$. When processing $u_1$, the algorithm tried to merge the cluster
  that $u_1$ was in with that of~$u_2$, 
 and we know that this attempt failed
  because we see that $u_1$ and
  $u_2$ are in different clusters at the end of the algorithm. This yields a
  contradiction with the fact that $u_1$ and $u_2$ end up in mergeable clusters
  at the end of the algorithm,
  given the order in which the algorithm considers possible merges. 
\end{proofsketch}

\begin{toappendix}
  \subparagraph*{Prefix order \& representative of a cluster.}
To compute the $k$-clustering we will rely on the \emph{prefix order}, which
is the order in which nodes are visited in a forest by a depth-first
traversal: when called on a node, the traversal processes the node,
then is called recursively on each child in the order over
children. Each cluster $C$ will have a \emph{representative}, which is
the node of $C$ that appears first in the prefix order.

\subparagraph*{Algorithm.}
The algorithm starts with each node in its own cluster and then uses a
standard depth-first traversal to saturate the clustering. To saturate
the clusters starting from a node $u$, we proceed recursively on its
first child and then on its next sibling. Then, we first try to merge
the cluster of $u$ with the cluster of its first child and then with
the cluster of its next sibling. By ``try to merge'', we mean that we
merge the two clusters if they are mergeable. By construction, the
resulting set of clusters is a $k$-clustering together with its forest
of clusters. We will now explain why the algorithm is correct, and
which data structures can be used to implement it with the claimed
linear time complexity.

\subparagraph*{Correctness of the algorithm.}

The clustering manipulated by our algorithm will always be a
$k$-clustering. For the correctness, we just need to verify that the
$k$-clustering obtained at the end of the algorithm is indeed a
saturated one.

Given a node $u$ and a forest~$\cT$, we denote $\cT_u$ the forest $\cT$
where we only kept nodes that are descendants of $u$ or descendants of
siblings of $u$ appearing after $u$; in other words $\cT_u$ is
the set of nodes reachable from $u$ repeatedly following a next
sibling or a first child relation.
Pay attention to the fact that $\cT_u$ is \emph{not} the subtree rooted
at~$u$ in~$\cT$, because it only contains the next siblings of~$u$.
Given a clustering on
$\cT$, we will say that it is \emph{saturated for
$\cT_u$} when we cannot merge clusters in~$\cT_u$: formally,
any two distinct clusters that each contain at least one node in~$\cT_u$
cannot be merged.

We will prove the following property inductively on the execution of the
algorithm:
\begin{claim}
  \label{clm:invar}
  When we run the algorithm on a forest $F$ and we are finished processing a
  node $u$, then the clustering
is saturated for $\cT_u$, the clusters of nodes in
$\cT_u$ can only contain nodes in $\cT_u$, and the clustering
  obtained at this point relative to the clustering obtained before we started
  processing $u$ is obtained by merging some clusters which are included in
  $\cT_u$.
\end{claim}
  Clearly, once we have inductively shown \cref{clm:invar}, we will have
  established that our algorithm produces a saturated clustering when it
  terminates, as
the forest $\cT$ is equal to $\cT_r$ for $r$ the first
root of $\cT$, and \cref{clm:invar} implies that $\cT_r$ is saturated.

When the algorithm processes a node $u$, it first processes recursively the
first child $\fc$ of~$u$ (if it exists) and then processes the next sibling
$\ns$
of $u$
(if it exists). By the induction
property (\cref{clm:invar}), we know that the clustering is saturated for
$\cT_{\fc}$ after processing $\fc$ and that it is saturated for
$\cT_{\ns}$ after processing $\ns$. Note that the recursive
processing on ${\ns}$ does not break the property on $\cT_{\fc}$
as processing $\ns$ only merges clusters included in~$\cT_{\ns}$
and the set of nodes in $\cT_{\ns}$ is
disjoint from the set of nodes in $\cT_{\fc}$. After
recursively processing $\fc$ and $\ns$, we thus know that the clustering
is saturated for both $\cT_{\ns}$ and
$\cT_{\fc}$. Furthermore, once we have recursively processed
$\fc$ and $\ns$, the only nodes of $\cT_u$ that can change
clusters while processing $u$ are nodes that end up in the same
cluster as $u$ (we are only merging clusters into the cluster of $u$).

By \cref{rem:mergeable}, to prove that a $k$-clustering is saturated,
it suffices to consider pairs of distinct clusters $C_1$ and $C_2$
that are LCRS-adjacent and show that they are not mergeable. Let us
consider $u$ such that $\cT_u$ contains two distinct
LCRS-adjacent clusters $C_1$ and $C_2$ after processing $u$.

Since $C_1$ and $C_2$ are LCRS-adjacent, when neither $C_1$ or $C_2$
contains $u$, they must both contain only nodes in $\cT_{\fc}$ or only
nodes in $\cT_{\ns}$ and the inductive property (\cref{clm:invar}) tells us that they are
not mergeable.

Let us now consider the case where $C_1$ contains $u$ (the case where
$C_2$ contains $u$ is symmetric) and there are three subcases. The first case is when
$C_2$ contains $\ns$, but this case is immediate as we have tried this exact merge and they were not
mergeable. The second case is when $C_2$ contains $\fc$, so that $C_1$ does not
contain any children of~$u$ in~$\cT$ except $u$ itself. For this second
case,
when the algorithm
started to process $u$ after having recursively processed $\fc$ and $\ns$, we had
that $u$ was in a singleton cluster which was not mergeable with $C_2$.
Let us show why this implies that $C_1$ is not mergeable with $C_2$.
Indeed, either $|C_2|=k$ and we conclude immediately that $C_1$ and $C_2$ are
not mergeable; or $C_2 \cup \{u\}$ contains two border nodes, namely, a border node
in~$C_2$ (because $C_2$ is a cluster so contains at most one border node) and a
border node in~$\{u\}$, namely, $u$. For this reason, $C_2$ does not contain all
children of~$u$. Thus, in $C_1 \cup C_2$, we have the same border node as
in~$C_2$, and $u$ is also a border node because $u$ has children which are not
in~$C_2$ and they are also not in~$C_1$. Thus, $C_1$ and $C_2$ are not mergeable.

The third case is when $C_1$ contains $u$ and $C_2$ contains
neither $\ns$ or $\fc$. Since $C_1$ contains $u$, it can be written as
$\{u\} \cup C_{\fc} \cup C_{\ns}$ where $C_{\fc}$ is either the cluster
in which $\fc$ ended up after processing it or the empty set (when the
merge with the cluster of $u$ failed) and similarly $C_{\ns}$ is either
the empty set or the cluster of $\ns$ after processing it.  Since $C_1$
and $C_2$ are LCRS-adjacent, let $u_1\in C_1$ and $u_2\in C_2$ be nodes
that are LCRS-adjacent. As $C_2$ contains neither $\fc$ nor $\ns$, we
cannot have $u=u_1$. Thus, we either have $u_1\in C_{\fc}$ or
$u_1\in C_{\ns}$. Let $e\in\{\fc,\ns\}$ be such that $u_1\in C_e$.
Note that $\cT_u = \{u\} \cup \cT_{\fc} \cup \cT_{\ns}$,
and since $C_2$ is a cluster of $\cT_u$ (in particular LCRS-connected)
and $C_2$ does not contain $u$ we have $C_2 \subseteq  \cT_{\ns}$ or
$C_2 \subseteq  \cT_{\fc}$, but $C_2$ contains $u_2$ which is adjacent to
$u_1 \in T_e$, so  we have 
$C_2 \subseteq \cT_e$. Further, we have $C_e = \cT_e \cap C_1$. 

Let us now show that $C_1$ is not mergeable with $C_2$. By induction hypothesis,
when the algorithm processed $e$, we knew that 
the clustering is saturated for
$\cT_{e}$, so
$C_e$ was not mergeable with
$C_2$. Now, either $|C_e\cup C_2| > k$ so that $|C_1 \cup C_2| > k$ (because
$C_e \subseteq C_1$) and we
immediately conclude that $C_1$ and $C_2$ are not mergeable; or $C_e \cup
C_2$ contains two border nodes. Let us now show in this latter case that every border node of $C_e
\cup C_2$ is a border node of~$C_1 \cup C_2$, which implies that $C_1$ and $C_2$
are indeed not mergeable.

Let us thus consider a border node $b$ of $C_2\cup C_e$. By definition,
$b$ has a child $c\not \in C_2 \cup C_e$. The node $c$ has to be in
$\cT_{e}$ and therefore it cannot be in $C_1\cup C_2$ if it is not in
$C_2\cup C_e$ (recall that $(C_1\cup C_2)\cap \cT_e = C_e\cup C_2$)
which proves that $b$ is also a border node of $C_1\cup C_2$.

\subparagraph*{A data structure to represent clusters.}
Clusters are LCRS-connected, we thus only need to remember for each
node whether it is in the same cluster as its first child or in the
same cluster as its next sibling. For this, we use two Boolean arrays
\texttt{mergedWithNextSibling} and \texttt{mergedWithFirstChild}. At the end
of the algorithm, we will read these arrays and build (in linear time,
with a single depth-first-search) the forest of clusters as well as,
for each node, an indication of the cluster to which this node belongs.

On top of this, we will also need extra data to check quickly whether
two clusters can be merged. For this, we use three arrays
\texttt{size}, \texttt{hasBorderNode} and \texttt{missingSibling}. For each
cluster $C$, letting $u$ be its representative node (recall that this is the
first node of~$C$ in the prefix order), then
\texttt{hasBorderNode[u]} will tell us whether $C$ contains a border
node \texttt{size[u]} will the size of $C$ (as an integer) and
\texttt{missingSibling[u]} tells us whether one of the roots of the
cluster has a next sibling which is not in the cluster.
Intuitively, the purpose of \texttt{missingSibling[u]} is to quickly
detect the situation described at the end of \cref{clm:mergeborder} where
merging the cluster of~$u$ with its parent makes the border node of the parent disappear.
The values of
\texttt{size[u]}, \texttt{missingSibling[u]} and
\texttt{hasBorderNode[u]} can be arbitrary for nodes \texttt{u} that are not the
representative of their cluster.

At the beginning of the algorithm, we initialize to false the arrays
\texttt{mergedWithNextSibling} and \texttt{mergedWithFirstChild}, and
the arrays \texttt{size}, \texttt{hasBorderNode}, and
\texttt{missingSibling} are initialized following the trivial
clustering: all sizes are 1, the nodes with border nodes are precisely
the internal nodes, and the nodes with missing siblings are precisely
the nodes having a next sibling (including roots which are not the
last root).

\subparagraph*{Merging clusters in $\bm{O(1)}$.}
Recall that our algorithm does not perform arbitrary merges between clusters,
and only tries to merge clusters $C_1$ and $C_2$ represented by $u_1$
and $u_2$ in one of two following cases.
The first case is when $u_2$ is the first child of $u_1$, in which case $C_1$
is just the single node $u_1 = u$ that the algorithm is processing.
The second case is when $C_1$ and $C_2$ are connected by a next-sibling
relation, in which case $C_1$ has $u_1 = u$ as its representative node.
In that case, note an easy consequence of the inductive property shown in the
correctness proof (\cref{clm:invar}): the cluster $C_1$ consists only of
descendants of~$u_1$. Indeed, we have just previously tried to merge $u$ with
its first child $u'$ if it exists: either the merge failed and then $C_1 = \{u\}$, or
it succeeded but then \cref{clm:invar} applied to~$u'$ ensures that the cluster
of $u'$
contained only nodes of $\cT_{u'}$ when we were done processing~$u'$, and the last point of
\cref{clm:invar} applied to nodes processed between $u'$ and $u$ ensures that
this property still holds when we start processing $u$.
So $C_1$ consists of $u_1 = u$
together with descendants and next siblings of~$u'$, i.e., descendants of~$u_1$.
This ensures that, in this second case, the representative node $u_2$ of~$C_2$
is the next sibling of $u_1$. 
Now, in both cases, it is easy to check whether the clusters
are mergeable:
\begin{itemize}
  \item Case 1: when $u_1$ has $u_2$ as first child, the clusters are
    mergeable when $C_2$ has size strictly less than $k$ and either
    $C_2$ has no border node or it contains all the children of $u_1$
    (i.e., when \texttt{missingSibling[u${}_2$]} is false, so the case at the end
    of \cref{clm:mergeborder} applies).
    To update \texttt{hasBorderNode[u]} and \texttt{missingSibling[u]}, note
    that the resulting cluster for $u_1 = u$ has a
    border node when either $C_2$ had a border node or when $C_2$ had no border
    node but had a missing sibling; and the resulting cluster has a missing
    sibling when $C_1$ had.
  \item Case 2: when $u_1$ has $u_2$ as next sibling, the clusters are
    mergeable if at most one of $C_1$ and $C_2$ has a border node 
    (this uses \cref{clm:mergeborder}) and
    if the total size is less than $k$.
    To update \texttt{hasBorderNode[u]} and \texttt{missingSibling[u]}, note
    that the resulting cluster for $u = u_1$ has a
    border node when one of $C_1$ and $C_2$ had, and it has a missing
    sibling when $C_2$ had.
\end{itemize}

Performing the actual merge is also easy: we simply update
\texttt{mergedWithNextSibling[u${}_1$]} or
\texttt{mergedWithFirstChild[u${}_1$]}, \texttt{hasBorderNode[u${}_1$]},
\texttt{size[u${}_1$]} and \texttt{missingSibling[u${}_1$]}.
Overall, these data structures make it possible to process each node in $O(1)$
in the algorithm.

\end{toappendix}

\myparagraph{Description of the main algorithm} We are now ready to
describe our algorithm for the dynamic evaluation problem for a fixed
forest algebra $(V,H)$. Given as input a $(V,H)$-forest $\cF$
with $n$ nodes (without a distinguished node), we want to maintain
the evaluation of~$\cF$ (which is an element in~$H$).  The first step
of the preprocessing is to compute in linear time an auxiliary data structure $\calS_{k}$ which
is used to maintain the evaluation of small forests under updates
in~$O(1)$, simply by tabulation.  More precisely, $\calS_{k}$ can be
used to solve the \emph{non-restricted} dynamic evaluation problem for
forests of size at most $k+1$, as follows:

\begin{toappendix}
  \subsection{Tabulation on Small Forests}
  All that remains is to re-state and show \cref{prp:tabulation}:
\end{toappendix}

\begin{propositionrep}
  \label{prp:tabulation}
Given a forest algebra $(V,H)$ there is a constant
$c_{V,H} \in \NN$ such that the following is true.
Given $k \in \NN$, we can compute in $O(2^{k\times c_{V,H}})$ a data structure
  $\calS_k$ that stores a sequence (initially empty) of
  $(V,H)$-forests $G_1, \dots, G_q$ and 
  supports the following:
\begin{itemize}
\item add($G$): given an $(V,H)$-forest $G$ with
  at most $k+1$ nodes, insert it into the sequence, taking time and space $O(|G|)$
\item relabel($i,n,\sigma$): given an integer $1 \leq i \leq q$,
  a node~$u$, and a label $\sigma \in H$ or $\sigma \in V$, relabel the
  node $u$ of~$G_i$ to $\sigma$, taking time $O(1)$ -- as usual we require that internal
    nodes have labels in~$H$ and at most one leaf has label in~$V$
\item eval($i$): given $1 \leq i \leq q$, return the evaluation of $G_i$, taking $O(1)$
\end{itemize}
\end{propositionrep}

\begin{proof}
Let us introduce the infinite graph $\Gamma$ where nodes are labeled
forests over the alphabet $\Sigma = H\cup V$. A node $\cF$ corresponds
to a forest with $\ell$ nodes and for each node $u$ in the forest $\cF$,
the \emph{index} of $u$ is the rank of the node in the prefix order of
the nodes of $\cF$. The edges that leave from the node $\cF$ are:
\begin{itemize}
\item for each $0\leq i \leq \ell$ and $\sigma\in\Sigma$ we have an
  edge $(add,i,\sigma)$ towards the forest $\cF'$, where this $\cF'$
  corresponds to the forest $\cF$ but with a new node as the last child of the node of
  index $i$ in $\cF$ that has label $\sigma$ (when $i=\ell$ this
  corresponds to add a new root appearing last in the prefix order);
\item for each $0\leq i < \ell$ and $\sigma\in\Sigma$ we have an edge
  $(lbl,i,\sigma)$ to the forest $\cF'$ corresponding to the same forest
  as $\cF$ but where the node of index $i$ in $\cF$ is labeled with
  $\sigma$.
\end{itemize}

In this forest we mark $\cF_\emptyset$ the node corresponding to the
empty forest.

This graph is infinite but we can consider $\Gamma_k$ the restriction of
$\Gamma$ to forests containing at most $k+1$ nodes. In this
restriction, each node is characterized by a ``shape'' (i.e. the
forest without labels) and a label for each node. There are less than
$4^{k+1}$ shapes. Indeed, this can be seen, e.g., by noticing that a forest without
labels of at most $k+1$ can be uniquely identified by a word
  of length $2k+2$ of opening and closing parentheses, formed of a
  well-parenthesized word of length at most $2k+2$ followed by a sequence of dummy opening parentheses to
  pad the length to $2k+2$. The number of words of length $2k+2$ on an alphabet
  of size~$2$ is $2^{2k+2}$, which is $4^{k+1}$, so this is also an upper bound on
  the number of shapes.

  Further, there are less than $|\Sigma|^{k+1}$ ways of labeling
each of them. This gives a total of less than $(4(|\Sigma|))^{k+1}$ nodes
in~$\Gamma_k$. For the number of edges, each node has at most $(2k+1)\times
|\Sigma|$ outgoing edges. Thus, for $k+1 = \left\lceil
\dfrac{\log(n)}{8 |\Sigma|} \right\rceil$, the number of nodes of $\Gamma_k$ is
upper bounded by:
  \[
    (4|\Sigma|)^{k+1} = \exp\left(\log(4|\Sigma|) \times \left\lceil
  \dfrac{\log(n)}{8 |\Sigma|} \right\rceil\right) \leq 
\exp\left(\log(4|\Sigma|) \times 
  \left(\dfrac{\log(n)}{8 |\Sigma|} + 1\right)\right)
\]
which by expanding yields
  \[
(4|\Sigma|)^{k+1} \leq
  n^{\dfrac{\log(4|\Sigma|)}{8 |\Sigma|}} \times 4|\Sigma|
  \leq 
  n^{1/2} \times 4|\Sigma|.
\]
  Thus, $\Gamma_k$ has 
$O(n^{1/2})$ nodes and we can compute in time $O(n)$ a representation
of $\Gamma_k$ that allows given a node $u$ to retrieve the neighbors using
the edges $(add,i,\sigma)$ and $(lbl,i,\sigma)$ in $O(1)$. On top of
that we can also pre-compute the evaluation of each forest
  of~$\Gamma_k$ -- we only compute the evaluation of well-formed forests,
  that is, $(V,H)$-forests where all internal nodes have a label in~$V$
  and at most one leaf has a label in~$V$.

To support the operation add$(\cF)$ for a forest $\cF$ that has less
than $k+1$ nodes, we can retrieve in $O(|\cF|)$ the node in~$\Gamma_k$
that represents $\cF$ by following edges $(add,i,\sigma)$ starting
from the node $\cF_\emptyset$. While doing that, we also store for each
node in~$\cF$ its index. Now, to support a relabeling update on a node
$u$ in~$\cF$, we just follow the edge $(lbl,i,\sigma)$ where $i$ is
the index of the node~$u$ in~$\cF$. Note that our scheme requires some data per forest
(the index for each node, and a pointer to a forest in~$\Gamma_k$) but
the graph $\Gamma_k$ is not modified therefore our scheme can support the
dynamic evaluation problem for multiple forests with constant-time
updates as claimed.
\end{proof}

Letting $k\coloneq \lfloor \log n /  c_{(V,H)} \rfloor$,
the first step of the preprocessing builds the data structure $\calS_k$ in time $O(n)$.
We then move on to the second step,
which is to compute a sequence of forests by recursively clustering in the following way. We
start by letting
$\cF_0 \coloneq \cF$ be the input $\Sigma$-forest.
Then, we recursively do the following. If $|\cF_i|=1$
the sequence stops at $\ell=i$. Otherwise, we compute a saturated
$k$-clustering $\equiv_i$ of $\cF_i$ using \cref{clm:algo}, and we let
$\cF_{i+1}$ be the forest of
clusters ${\cF_i}^{\equiv_i}$. We will later show that this 
takes time $O(n)$ overall.

We continue with a third preprocessing step that computes the evaluation of each cluster at
each level: again we will later argue that this takes 
$O(n)$.
More precisely, we consider all the clusters $C$ of $\cF_0$
and add the sub-forest $\cF_0^C$ of $\cF_0$ induced by each~$C$
to~$\calS_k$, obtaining their evaluation in time $O(|\cF_0|)$. We
use the result of the evaluation as labels for the corresponding nodes in
$\cF_1$.
Then we add the sub-forests induced by all the clusters
of~$\cF_1$ to~$\calS_k$ to obtain their evaluation and to label
$\cF_2$.
We continue 
until we have the evaluation of $\cF_\ell$.
Note that none of the $(V,H)$-forests $\cF_i$ has a distinguished leaf:
indeed, the only place where we perform non-restricted dynamic evaluation is in
\cref{prp:tabulation}, i.e., on the sub-forests induced by the
clusters and added to~$\calS_k$.
Further note that all these induced sub-forests 
have at most
$k+1$ nodes by definition of a $k$-clustering
(the $+1$ comes from the $\square$-labeled leaf which may be added for the border node in
\cref{def:clusteval}).
Now, by \cref{clm:correct} at each step, we have that 
the evaluation of $\cF_\ell$ is equal to the evaluation
of the input $(V,H)$-forest $\cF_0$. This is the answer we need to
maintain, and 
this concludes 
the preprocessing.

Let us now explain how we recursively handle relabeling updates. To
apply an update to the node $u$ of the input forest~$\cF_0$, we
retrieve its cluster $C$ and use $\calS_k$ to apply this update to~$u$
to the induced sub-forest $\cF_0^C$ and retrieve its new evaluation,
in~$O(1)$. This gives us an update to apply to the node~$C$ of the
forest of clusters~$\cF_1$. We continue like this over the successive
levels, until we have an update applied to the single node of
$\cF_\ell$, which again by \cref{clm:correct} is the desired
answer. The update is handled overall in time $O(\ell)$, so let us
bound the number $\ell$ of recursion steps. At every level $i<\ell$ we
have $ |\cF_{i+1}|\leq \lceil |\cF_{i}|\times (c_{V,H}/k)\rceil$ by
\cref{prp:saturated} and $|F_i| \geq k $, so we have $|\cF_i| \leq n\times
((c_{V,H}+1)/k))^i$ and therefore $\ell = O(\log n / \log k) = O(\log
n / \log \log n)$
given that $k= \lfloor \log n / c_{V,H} \rfloor$.

The only remaining point is to bound the total time needed in the preprocessing.
The data structure $\calS_k$ is computed in linear time in the first step. Then, we spend linear time in each 
$\cF_i$ to compute the $k$-clusterings at each level in the second step, and we again spend
linear time in each $\cF_i$ to feed the induced sub-forests of all the
clusters of~$\cF_i$ to $\calS_k$ in the third step.
Now, it suffices to observe that
the size of each $\cF_i$ decreases exponentially with~$i$; for
sufficiently large $n$ and $k$ we have $k \geq 2$ so $|\cF_i| \leq
|\cF_0|/2^i$ and the total size of the $\cF_i$ is
$O(|\cF_0|)$. This ensures that the total time taken by the preprocessing across
all levels in the second and third steps is in~$O(n)$, which
concludes.

\section{Dynamic Membership to Almost-Commutative Languages in $O(1)$}
\label{sec:zgupper}
We have shown in \cref{sec:lll} our general upper bound (\cref{thm:upper}) on
the dynamic membership problem to arbitrary forest languages.
In this section, we study how we can show more favorable bounds on the
complexity on dynamic membership when focusing on restricted families of 
languages. More precisely, in this section, we define a class of regular
forest languages, called \emph{almost-commutative} languages, and show that the
dynamic membership problem to such languages can be solved in~$O(1)$.  We will
continue studying these languages in the next section to show that
non-almost-commutative languages conditionally cannot be maintained in constant
time when assuming the presence of a neutral letter.

\subparagraph*{Defining almost-commutative languages.}
To define almost-commutative languages, we need to define the
subclasses of \emph{virtually-singleton}  languages and
\emph{regular-commutative} languages. Let us first define virtually-singleton
languages via the operation of \emph{projection}:

\begin{definition}
  \label{def:virtsing}
  The \emph{removal} of a node~$u$ in a $\Sigma$-forest $F$ means 
replacing $u$ by the (possibly empty) sequence of its children.
\emph{Removing} a subset of nodes of~$F$ is then defined in the expected way;
note that the result does not depend on the order in which nodes are removed.

  For $\Sigma' \subseteq \Sigma$ a subalphabet, given a $\Sigma$-forest $F$, the 
  \emph{projection} of \(F\) over \(\Sigma'\) is the forest
\(\pi_{\Sigma'}(F)\)
  obtained from \(F\) when removing all nodes that are labeled by a letter 
of \(\Sigma \setminus \Sigma'\).

  A forest language \(\LL\) over $\Sigma$ is \emph{virtually-singleton} if there exists a
  subalphabet \(\Sigma'\subseteq \Sigma\) and a $\Sigma'$-forest \(F'\) such
  that \(\LL\) is the set of forests whose projection over \(\Sigma'\) is \(F'\).
\end{definition}

Note that virtually-singleton languages are always regular: a forest automaton
can read an input forest, ignoring nodes with labels in $\Sigma \setminus
\Sigma'$, and check that the resulting forest is exactly the fixed target
forest~$F'$.

Let us now define \emph{regular-commutative} languages as the
\emph{commutative} regular forest languages, i.e.,
membership to
the language can be determined from the \emph{Parikh image}:

\begin{definition}
  \label{def:regcom}
The \emph{Parikh image} of a $\Sigma$-forest \(F\) is the vector \(v\in
  \N^{\Sigma}\) such that for every letter \(a \in \Sigma\),
the component \(v_{a}\) is the number of nodes labelled by \(a\) in \(F\).

  A forest language \(\LL\) is \emph{regular-commutative} if it is regular
  and there is a set \(S\subseteq \N^{\Sigma}\) such that
  \(\LL\) is the set of forests whose Parikh image is in \(S\).
\end{definition}

We can now define 
\emph{almost-commutative} forest languages from these two
classes:

\begin{definition}
  A forest language \(\LL\) is \emph{almost-commutative} if it is a finite Boolean combination of regular-commutative and
  virtually-singleton languages.
\end{definition}

Note that almost-commutative languages are always regular, because regular
forest
languages are closed under Boolean operations.
Further, in a certain sense, almost-commutative languages
generalize all word languages with a neutral letter that enjoy constant-time
dynamic membership. Indeed, such languages are known
by~\cite{amarilli2021dynamic} and \cite[Corollary~3.5]{amarilli2023locality} to
be described by regular-commutative conditions and the presence of specific
subwords on some subalphabets.
Thus, letting $L$ be such a word language,
we can define a forest language $L'$
consisting of the forests $F$ where the nodes of~$F$ form a word of~$L$ (e.g.,
when taken in prefix order), and $L'$
is then almost-commutative.
As a kind of informal converse, given a $\Sigma$-forest $F$, we can represent it as a
word with opening and closing parentheses (sometimes called the \emph{XML
encoding}), and the set of such representations for an almost-commutative forest
language $L'$ will intuitively form a word language $L'$ that enjoys
constant-time dynamic membership except for the 
(non-regular) requirement that parentheses are balanced.

We also note that we can effectively
decide whether a given forest language is almost-commutative, as will follow
from the algebraic characterization in the next section:

\begin{propositionrep}
  The following problem is decidable:
  given a forest automaton $A$, determine
  whether the language accepted by $A$
  is almost-commutative.
\end{propositionrep}

\begin{proof}
  We use the equivalence between almost-commutative languages and ZG forest
  algebras (\cref{thm:almost-comm_zg}) which is shown in \cref{sec:zglower}. So
  it suffices to show that we can compute the syntactic forest algebra and test
  whether its vertical monoid satisfies the ZG equation. For this, we use the
  process described in~\cite{bojanczyk2008forest} to compute the forest algebra,
  and then for every choice of elements $v$ and $w$ of the vertical monoid we
  compute the idempotent power $v^\omega$ of~$v$ and then check whether
  $v^{\omega+1}w = wv^{\omega+1}$.
\end{proof}

\subparagraph*{Tractability for almost-commutative languages.}
We show the main result of this section:

\begin{theoremrep}\label{thm:almost_cor_in_RAM}
  For any fixed almost-commutative forest language $L$, the dynamic membership problem
  to~$L$ is in $O(1)$.
\end{theoremrep}

\begin{proofsketch}
This result is shown by proving the claim for regular-commutative languages and
virtually-singleton languages, and then noticing that tractability is preserved
under Boolean operations, simply by combining data structures for the
  constituent languages.

  For regular-commutative languages, we maintain the Parikh image as a vector in constant-time per update, and we
  then easily maintain the forest algebra element to which it corresponds,
  similarly to the case of monoids (see \cite{skovbjerg1997dynamic} or
  \cite[Theorem 4.1]{amarilli2021dynamic}).
  
  For virtually-singleton languages, we use doubly-linked lists
  like in \cite[Prop.~4.3]{amarilli2021dynamic} to maintain, for each
  letter $a$ of the subalphabet~$\Sigma'$, the
  unordered set of nodes with label~$a$. This allows us to
  determine in constant-time whether the Parikh image of the input forest
  restricted to~$\Sigma'$ is correct: when this holds,
  then the doubly-linked lists have constant size and we can use them to recover
  all nodes with labels in~$\Sigma'$. With constant-time
  reachability queries, we can then test if these nodes achieve the requisite
  forest $F'$ over~$\Sigma'$.
\end{proofsketch}

\begin{toappendix}
This result is shown by proving the claim for regular-commutative languages and
virtually-singleton languages, and then noticing that tractability is preserved
under Boolean operations.

\begin{lemmarep}\label{lem:reg_com_in_RAM}
  For any fixed regular-commutative forest language $L$,
  the dynamic membership problem to~$L$ is in~$O(1)$.
\end{lemmarep}

\begin{proof}
  The proof follows \cite[Theorem 4.1]{amarilli2021dynamic}, and works in the
  syntactic forest algebra $(V,H)$ for the language~$L$.
  Let $n$ be the number of nodes of the input tree. In the preprocessing phase,
  for each $a \in \Sigma$ and for each $0 \leq i \leq n$, we precompute the
  element $h(a,i)$ of~$H$ which is the image by~$\mu$ of a forest consisting of $i$ roots
  labeled~$a$ and having no children: this can be achieved by repeated
  composition of $\mu(a)$ in~$H$. 
  Now, for dynamic evaluation, we easily maintain the Parikh image of the
  forest~$F$
  under updates in constant time per update. As $L$ is commutative, the
  membership of $F$ to~$L$ can be decided by testing the membership to~$L$ of
  the forest  $F'$ having only roots with no children, with the number of occurrences of
  each letter described by the Parikh image. The image
  of~$F'$ by~$\mu$ can be computed in constant time by composing the
  precomputed elements $h(a,i)$, from which we can determine membership of~$F'$,
  and hence of~$F$, to~$L$.

  As in \cite[Theorem 4.1]{amarilli2021dynamic}, we note that the precomputation
  can be avoided because regular-commutative languages must be imposing
  ultimately periodic conditions on the Parikh image.
\end{proof}

\begin{lemmarep}\label{lem:virt_sing_in_RAM}
  For any fixed virtually-singleton forest language $L$,
  the dynamic membership problem to~$L$ is in~$O(1)$.
\end{lemmarep}

\begin{proof}
   Let \(\LL\) be a virtually-singleton language, with subalphabet \(\Sigma'\) and
   $\Sigma'$-forest \(F'\) of size~\(k\). Given the input forest $F$, we precompute in time
   $O(|F|)$ an ancestry data structure allowing us to answer the following reachability
   queries in constant time: given two nodes $u, u' \in F$, decide if $u$ is an
   ancestor of~$u'$. This can be achieved by doing a depth-first
   traversal of~$F$ and labeling each node $u$ with the timestamp $p_u$ at which
   the traversal enters node $u$ and the timestamp $q_u$ at which the traversal
   leaves node $u$. Then $u$ is an ancestor of $u'$ iff we have $p_u \leq
   p_{u'}$ and $q_{u'} \leq q_u$. Note that this data structure depends only on
   the shape of~$F$, so it is only computed once during the preprocessing and is
   never updated.

   We then use the doubly-linked list data structure of
   \cite[Proposition 4.3]{amarilli2021dynamic}. Namely, for each letter $b \in
   \Sigma'$, we compute a doubly-linked list $L_b$ containing pointers to every occurrence
   of~$b$ in the forest~$F$ (in no particular order); and for every node $u$ of~$F$ labeled with~$b$ we
   maintain a pointer $\phi_u$ to the list element that represents it in~$L_b$. We can
   initialize this data structure during the linear-time preprocessing by a
   traversal over $F$ where we populate the doubly-linked lists. Further, we can
   maintain these lists in $O(1)$ at every update. When a node $u$ loses a label
   $b \in \Sigma'$, then we use the pointer $\phi_u$ to locate the list element
   for~$u$ in~$L_b$, we remove the list element in constant time from the
   doubly-linked list, and we clear the pointer $\phi_u$. When a node $u$ gains a
   label $b \in \Sigma'$, then we append
   $u$ to~$L_b$ (e.g., at the beginning), and set $\phi_u$ to point to the newly
   created list item in~$L_b$.

   We now explain how we determine in $O(1)$ whether the current forest belongs
   to~$L$. First, as $|F'|$ is constant, we know that for each $b \in \Sigma'$, we can determine in $O(1)$ whether $L_b$ contains
   exactly 
   $|F'|_b$ elements -- note that we do not need to traverse all of $L_b$
   for this, and can stop early as soon as we have seen $|F'|_b+1$ elements. If
   the test fails for one letter $b \in \Sigma'$, then we know that the current forest
   does not belong to~$L$, because its projection to $\Sigma'$ will not have the right
   Parikh image.

   Second, if all tests succeed, we can retrieve in $O(1)$ from the lists $L_b$ for $b
   \in \Sigma'$ the
   occurrences of letters of~$\Sigma'$ in~$F$, giving the list of 
   nodes $n_1, \ldots, n_k$ of~$F$ with a label in~$\Sigma'$: remember that their number is
   exactly $k = |F'|$ from our test above. These are the nodes of
   the forest $\pi_{\Sigma'}(F)$.

   Now, for any pair of nodes $n_i, n_j$ for two distinct $i, j \in \{1,
   \ldots, k\}$, we can
   query the ancestry data structure of~$F$ computed at the beginning of
   the proof to know what are the edges of the forest
   $\pi_{\Sigma'}(F)$: there are at most $k^2$ queries, i.e., a constant number of
   queries, and each query is answered in constant time. (To be precise,
   the queries give us the descendant relationship of the forest, from
   which the child relationship can be easily computed.)
   Further, again in time
   $O(k^2)$ which is again a constant, we can order the sibling sets of
   $\pi_{\Sigma'}(F)$ using the timestamps precomputed in the first paragraph: for any two nodes $n_i$ and $n_j$ that are siblings
   in~$\pi_{\Sigma'}(F)$ we have that $n_i$ is a preceding sibling of $n_j$ if
   and only if the timestamp $p_{n_i}$ is less than the
   timestamp $p_{n_j}$. (Again, to be precise, this gives us the transitive
   closure of the order relationship on sibling sets, from which the
   previous/next sibling relationships are easily computed.)

   Now that the forest
   $\pi_{\Sigma'}(F)$ is fully specified, we can check if it is isomorphic to $F'$: this is in time
   $O(1)$ because both forests have size $k$ which is a constant. From this, we 
   conclude from this whether $F \in L$ or $F \notin L$.
 \end{proof}

From \cref{lem:reg_com_in_RAM}
and \cref{lem:virt_sing_in_RAM},
we can conclude
the proof of \cref{thm:almost_cor_in_RAM}: 
for any almost-commutative language $L$, we can prepare 
data structures for constant-time dynamic membership to its constituent
almost-commutative and virtually-singleton languages, and the Boolean
membership informations maintained by these data structures can be combined in
constant-time via Boolean operations to achieve a constant-time dynamic
membership data structure for~$L$. This concludes the proof of
\cref{thm:almost_cor_in_RAM}.

\end{toappendix}

\section{Lower Bound on Non-Almost-Commutative Languages with Neutral Letter}
\label{sec:zglower}
We have introduced in the previous section the class of \emph{almost-commutative
languages}, and showed that such languages admit a constant-time dynamic
membership algorithm (\cref{thm:almost_cor_in_RAM}). In this section, we show
that this class is tight: non-almost-commutative regular forest languages cannot
enjoy constant-time dynamic membership when
assuming the presence of a
neutral letter and when assuming
the hardness of the
\emph{prefix-$U_1$} problem from~\cite{amarilli2021dynamic}. We first present
these hypotheses in more detail, and state the lower bound that we show in this section
(\cref{thm:lb}).
Second, we present the algebraic characterization of almost-commutative regular
languages on which the proof of \cref{thm:lb} hinges: they are precisely the regular languages
whose syntactic forest algebra is in a class called ZG.
Third, we sketch the lower bound showing that dynamic membership is
conditionally not in $O(1)$ for languages with a neutral letter whose syntactic forest algebra
is not in ZG, and conclude.

\myparagraph{Hypotheses and result statement}
The lower bound shown in this section is conditional to a computational hypothesis
on the \emph{prefix-$U_1$ problem}. In this problem, we are given as input a word $w$ on the alphabet $\Sigma =
\{0,1\}$, and we must handle substitution updates to~$w$ and queries where we
are given $i \in \{1, \ldots, |w|\}$ and must return whether the prefix of $w$ of
length $i$ contains some occurrence of~$1$. In other words, we must
maintain a subset of integers of $\{1, \ldots, |w|\}$ under insertion and
deletion, and handle queries asking whether an input $i \in \{1, \ldots, |w|\}$
is greater than the current minimum: note that this like a priority queue
except
we can only compare the minimum to an
input value but not retrieve it directly. We will use the following conjecture from~\cite{amarilli2021dynamic}
as a hypothesis for our lower bound:

\begin{conjecture}[\protect{\cite[Conj.\
  2.3]{amarilli2021dynamic}}]\label{conj:prefix_uone}
  There is no data structure solving the prefix-$U_1$ problem in
  \(\bigO(1)\) time per operation in the RAM model with unit cost and
  logarithmic word size.
\end{conjecture}

Further, our lower bound in this section will be shown for regular forest languages $L$ over
an alphabet $\Sigma$ which are assumed to feature a so-called \emph{neutral
letter} for~$L$.
Formally,
a letter $e \in \Sigma$ is \emph{neutral} for~$L$ if, for every
forest $F$, we have $F \in L$ iff $\pi_{\Sigma \setminus \{e\}}(F) \in L$, where
$\pi$ denotes the projection operation of \cref{def:virtsing}.
In other words, nodes labeled by $e$ can be removed without affecting the
membership of~$F$ to~$L$.

We can now state the lower bound shown in this section, which is our main
contribution:

\begin{theorem}
  \label{thm:lb}
  Let \(\LL\) be a regular forest language featuring a neutral letter.
  Assuming \cref{conj:prefix_uone},  we have that $L$ has dynamic membership in
  $\bigO(1)$ iff $L$ is almost-commutative.
\end{theorem}

\myparagraph{Algebraic characterization of almost-commutative languages}
The proof of \cref{thm:lb} hinges on an algebraic characterization of the
almost-commutative languages. Namely, we will define a class of forest algebras
called \emph{ZG}, by imposing the \emph{ZG equation}
from~\cite{amarilli2021dynamic,amarilli2023locality}: 

\begin{definition}
  \label{def:zg}
  A monoid $M$ is in the class \emph{ZG} if it satisfies the \emph{ZG equation}:
  for all $v, w \in M$ we have
    \(v^{\omega+1}w = wv^{\omega+1}\).  %
  A forest algebra $(V, H)$ is in the class \emph{ZG} if its vertical monoid is
  in ZG: we call it a \emph{ZG forest algebra}.
\end{definition}

The intuition for the ZG equation is the following. Elements of the form
$x^{\omega+1}$ in a monoid are called \emph{group elements}: they are precisely
the elements which span a subsemigroup (formed of the elements $x^{\omega+1},
x^{\omega+2}, \ldots$) which has the structure of a group (with $x^\omega =
x^{2\omega}$ being the neutral element). Note that the group is always a cyclic
group: it may be the trivial group, for instance in an aperiodic monoid all such
groups are trivial, or in arbitrary monoids the neutral element always spans
the trivial group. The equation implies that all group elements of the monoid
are \emph{central}, i.e., they commute with all other elements.

The point of ZG forest algebras is that they correspond to almost-commutative
languages:

\begin{toappendix}
  \subsection{Proof of \cref{thm:almost-comm_zg}}
  \label{apx:zgproof}
\end{toappendix}

\begin{theoremrep}\label{thm:almost-comm_zg}
  A regular language \(\LL\) %
  is almost-commutative if and only if its syntactic forest algebra is in ZG.
\end{theoremrep}

\begin{toappendix}
  The proof is split in three parts, spanning the next sections. First, we
  derive some equations that hold on ZG forest algebras. Second, we show the
  forward direction, which is easy. Third, we show the backward direction, which
  is more challenging and uses the equations derived in the first section.
\subsubsection{Equations Implied by (ZG)}
\label{apx:equations}
As a preliminary step, we prove that the equation (ZG) from \cref{def:zg}
on a forest algebra implies
several other equations.

Throughout the
section, we fix a forest algebra $(V,H)$.
For $i \in \mathbb{N}$, we write $i \cdot h$ for an element $h \in H$ to mean
$h$ composed with itself $i$ times, like exponentiation; we write $v^i$ for $v
\in V$ to mean the same for elements of~$V$.
Further, we will write $\omega \cdot h$ for an element $h \in H$ to mean the idempotent
power of~$h$, and write as usual $v^\omega$ for the idempotent power of $v \in
V$.
For \(v\in V\) and \(k\in \mathbb{N}\), we simply define \(v^{\omega +k}\) to be \(v^{\omega} \cdot v^{k}\).
We want to extend this definition to make sense of \(v^{\omega -k}\).
Let \(q\) an integer such that \(v^{q} = v^{\omega}\).
As every multiple of \(q\) satisfies the equation, we pick \(q'\geq q+k\) such that \(v^{q'}=v^{\omega}\).
Then \(v^{\omega-k}\) is defined as \(v^{q'-k}\).
This definition ensures that for any two relative integers \(k_{1}\) and \(k_{2}\), we have \(v^{\omega +k_{1}}\cdot v^{\omega + k_{2}} = v^{\omega\cdot k_{1}+k_{2}}\).
We define similarly \((\omega +k)\cdot h\) for \(h\in H\) and \(k\in\mathbb{Z}\).

\subparagraph*{Centrality of other forms of elements.}
First, we remark that (ZG) implies the centrality of other elements, for the
following reason (also observed as \cite[Claim~2.1]{amarilli2023locality}):

\begin{claim}
  \label{clm:omegak}
  For any monoid $M$, element $v \in M$, and integer $k \in \mathbb{Z}$,
  we have $v^{\omega+k} = (v^{\omega+k})^{\omega+1}$. Hence, if $M$ satisfies
  the equation (ZG), then $v^{\omega+k}$ is central.
\end{claim}

So in particular for $k=0$ we know that all idempotents of a ZG monoid are
central.

\begin{proof}[Proof of \cref{clm:omegak}]
We simply have:
\[ (v^{\omega+k})^{\omega+1} = v^{\omega}\cdot v^{\omega+k} = v^{\omega+k}.
  \qedhere \]
\end{proof}

\subparagraph*{ZG equation on the horizontal monoid.}
Let us then remark the following fact, which is stated in \cite[Fact
2.32]{barloy2024complexity} but which we re-state here to match our definitions. In this
statement, we say that $N$ is a \emph{submonoid} of~$M$ if there is an injective
morphism from~$N$ to~$M$:

 \begin{fact}\label{fact:H_submonoid_V}
   For every forest algebra \((V,H)\), we have that \(H\) is a submonoid of \(V\).
 \end{fact}

 \begin{proof}
   We define a morphism \(\mu\) from \(H\) to \(V\) by \(\mu(h) = h \oplus \square\) for \(h\in H\).
   It is indeed a morphism: for \(h,h',g\in H\), we have on the one hand that \(\mu(h+h')\odot g = ((h+h')\oplus\square)\odot g = h+h'+g\) by the Mixing axiom.
   On the other hand, we have that \((\mu(h)\cdot \mu(h'))\odot g = (h\oplus\square) \odot [(h'\oplus \square) \odot g] =(h\oplus\square) \odot (h'+g) = h+h'+g \)
   by applying successively the Action axiom then twice the Mixing one.
   We can then conclude with the Faithfulness axiom that \(\mu(h+h')=\mu(h)\cdot \mu(h')\).

   We finally need to prove that this morphism in injective.
   Let \(h,h'\in H\) such that \(\mu(h)=\mu(h')\).
   Hence \((h\oplus\square)\odot \epsilon =(h'\oplus\square)\odot \epsilon  \) and thus, by the Mixing axiom, \(h = h+ \epsilon = h'+\epsilon = h'\).
 \end{proof}

Thanks to \cref{fact:H_submonoid_V}, we know that the horizontal monoid
of a ZG forest algebra is in ZG as well, namely:

\begin{claim}
  \label{clm:zgh}
  For \((V,H)\) a ZG forest algebra and \(h,g\in H\), we have:
\begin{equation}
  \label{eq:zgh}
(\omega+1)\cdot h + g = g + (\omega+1)\cdot h . \tag{ZGh}
\end{equation}
\end{claim}

\begin{proof}
Indeed, letting $h$ be an arbitrary element of~$H$, the image in~$V$ of
$(\omega+1)\cdot h$ by the injective morphism from~$H$ to~$V$ is of the form
$v^{\omega+1}$, for $v$ the image of~$h$ by the morphism.
This stands thanks to the fact that a morphism and the idempotent power commutes.
Thus, the (ZG) equation on~$V$ implies that $v^{\omega+1}$ is central in~$V$.
Thus, for any $h$ in~$H$, we see that the left-hand side and right-hand side of
\cref{eq:zgh} evaluate respectively in~$V$ to $v^{\omega+1} w$ and $w
v^{\omega+1}$, for $w$ the image of~$h$ by the morphism; by centrality of
$v^{\omega+1}$ in~$V$ they are the same element, and the injectivity of the
morphism implies that the left-hand side and right-hand side of \cref{eq:zgh}
are equal.
\end{proof}

We will also use a known result about the ZG equation on monoids
(\cite[Lemma 3.8]{amarilli2023locality}), which claim that the idempotent powers
distribute. Instantiated to the setting of the horizontal and vertical monoids
of ZG forest algebras, which both satisfy the ZG equation by
\cref{clm:zgh}, we immediately get the following from \cite[Lemma 3.8]{amarilli2023locality}:

\begin{claim}
  \label{clm:distomega}
For every \(h,g\in H\) and \(v,w\in V\) in a ZG forest algebra we have:
\begin{align*}
(vw)^{\omega} &= v^{\omega}w^{\omega} , \tag{DISTv}  \\
\omega \cdot (h+g) &= \omega\cdot h + \omega\cdot g .\tag{DISTh}
\end{align*}
\end{claim}

The equation (ZG) also gives interesting interactions between the vertical and
horizontal monoids. Intuitively, forests that are group elements of the
horizontal monoid can be taken out of any context; and we can take any forest
out of contexts that are group elements of the vertical monoid:

\begin{lemma}\label{lem:OUTh_OUTv}
  Let \((V,H)\) be a ZG forest algebra.
  It satisfies the following equations, for every \(h\in H\) and \(v\in V\):
  \begin{align*}
    v\odot((\omega +1)\cdot h) &= v\odot\epsilon + (\omega+1)\cdot h, \tag{OUTh} \\
    v^{\omega +1}\odot h &= v^{\omega+1}\odot\epsilon + h \tag{OUTv}.
  \end{align*}
  where $\epsilon$ denotes the empty forest.
\end{lemma}

\begin{proof}
  Let \(w = \square\oplus h \), which is an element of \(V\).
  For every \(q\in\N\),  we have that \(w^{q} = \square\oplus  q\cdot h\) by the Mixing axiom.
  If \(w^{q}\) is idempotent, then \(w^{2q}\odot\epsilon = w^{q}\odot\epsilon\) and so \(q\cdot h\) is idempotent as well.
  It implies that \(w^{\omega}=\square\oplus  \omega\cdot h\), because for \(q\) such that \(w^{\omega}=w^{q}\), there is \(w^{\omega} = \square \oplus q \cdot h = \square\oplus \omega\cdot h\).
  Thus \(w^{\omega+1}=\square\oplus (\omega+1)\cdot h\).
  We can apply (ZG) on \(v\) and \(w\): \(vw^{\omega+1} = w^{\omega+1}v\).
  This rewrites into \(v \cdot (\square\oplus(\omega+1)\cdot h ) = v\oplus(\omega+1)\cdot h \).
  Applying the empty forest \(\epsilon\) to both sides gives (OUTh).

  Now, for (OUTv), with the same \(w=\square\oplus h\),
  we apply (ZG) to \(v\) and \(w\) to get: \(v^{\omega+1}w = wv^{\omega+1}\).
  This rewrites to \(v^{\omega+1}\cdot (\square\oplus h) = v^{\omega+1} + h \).
  Applying the empty forest \(\epsilon\) to both sides gives (OUTv).
\end{proof}

We moreover obtain an equation that says that vertical idempotents are horizontal idempotents as well.

\begin{lemma}\label{lem:IDv}
  Let \((V,H)\) be a ZG forest algebra.
  For every \(v\in V\) and \(i,j\in \N\), we have that
  \[ v^{\omega+i}\odot\epsilon + v^{\omega+j}\odot\epsilon  = v^{\omega+i+j}\odot\epsilon . \]
  In particular,
  \[ v^{\omega}\odot\epsilon  + v^{\omega}\odot\epsilon  = v^{\omega}\odot\epsilon .\tag{IDv} \]
\end{lemma}

\begin{proof}
  Let \(v\in V\).
  We write \(v^{\omega+i+j}\odot\epsilon  = v^{\omega+i}\odot(v^{\omega+j}\odot\epsilon )\) and we apply
  (OUTv) (with \(v^{\omega+j}\odot\epsilon\) playing the role of \(h\), relying on
  \cref{clm:omegak}).
  This gives \(v^{\omega+i}\odot(v^{\omega+j}\odot\epsilon) = v^{\omega+i}\odot\epsilon + v^{\omega+j}\odot\epsilon\).
  The ``in particular'' part comes from the special case \(i=j=0\).
\end{proof}

The last important equation is an equation that draws a bridge between horizontal and vertical idempotent powers.

\begin{lemma}\label{lem:flat_eq}
  Let \((V,H)\) be a ZG forest algebra.
  For every \(h\in H\) and \(v\in V\), we have that:
  \[ \omega\cdot (v \odot h) = v^{\omega}\odot\epsilon + \omega \cdot h. \tag{FLAT} \]
\end{lemma}

\begin{proof}
  Let $w = h\oplus\square$, which is an element of~$V$.
  Like at the beginning of the proof of \cref{lem:OUTh_OUTv},
  it stands that \(w^{\omega} = \omega\cdot h \oplus\square\),
  and \(v\odot h = (v\cdot w) \odot\epsilon \).

  The goal is to prove that prove that both sides of the equation (FLAT)
  are equal to \((v\cdot w)^{\omega}\odot\epsilon + \omega\cdot (v\odot h)\).

  On one hand, for the right-hand side, we have:
  \begin{align*}
    v^{\omega}\odot\epsilon + \omega\cdot h &= v^{\omega}\odot(\omega\cdot h) \tag{by (OUTh)} \\
    &= (v^{\omega}\cdot w^{\omega}) \odot\epsilon \\
    &= (vw)^{\omega}\odot\epsilon \tag{by (DISTv)} \\
    & =(vw)^{\omega-1}\oplus(v\oplus h)) \tag {because $h = w\oplus\epsilon$}\\
    &= (vw)^{\omega-1}\odot\epsilon + v\odot h. \tag{by (OUTv)}
  \end{align*}

  We used~\cref{clm:omegak} implicitly in the last equality.
  Repeating the process above, for all \(k\in\N\), we can show the equation ($\star$): \( v^{\omega}\odot\epsilon + \omega\cdot h = (v\cdot w)^{\omega-k}\odot\epsilon + k\cdot (v\odot h).\)
  We will apply this for $k \coloneq m$
  where \(m\) is a multiple of the idempotent powers of every element of both \(H\) and \(V\).
  With this value, we also have \(m\cdot (v\odot h) = \omega\cdot (v\odot h)\) and \((v\cdot w)^{\omega-m} = (v\cdot w)^{\omega}\).
  Hence, plugging $k \coloneq m$ into ($\star$), we get: \(v^{\omega}\odot\epsilon + \omega\cdot h = (vw)^{\omega}\odot\epsilon + \omega\cdot v(h)\).

  \medskip

  On the other hand, for the left-hand side, we have:
  \begin{align*}
    \omega\cdot (v\odot h) &=  (v\cdot w)\odot\epsilon + (\omega-1)\cdot (v\odot h) \\
    &= (v\cdot w)\odot((\omega-1)\cdot (v\odot h)) \tag{by (OUTh)}.
  \end{align*}

  Repeating the process, for all \(k\in\N\), \(\omega\cdot (v\odot h) = (v\cdot w)^{k}\odot((\omega-k)\cdot (v\odot h))  \).
  As previously, with \(k\) set to a multiple of the idempotent powers of every element of both \(H\) and \(V\), we have that
  \begin{align*}
    \omega\cdot (v\odot h) &= (v\cdot w)^{\omega}\odot(\omega \cdot (v\odot h)) \\
            &= (v\cdot w)^{\omega}\odot\epsilon + \omega\cdot (v\odot h) \tag{by (OUTh)}
  \end{align*}

  We have proved that both sides of the equation (FLAT) are equal to the desired
  value. This concludes the proof.
\end{proof}

  \subsubsection{Forward Direction}

The first direction to show \cref{thm:almost-comm_zg} is the forward direction: 
every almost-commutative language has a syntactic forest algebra in
this class:

\begin{lemma}\label{lem:almost_comm_satisfies_ZG_words}
  Let \(\LL\) be an almost-commutative language and \((V,H)\) be its syntactic forest algebra.
  We have that \((V,H)\) satisfies the equation (ZG).
\end{lemma}

For this we will need to use the fact that, when the syntactic forest algebras
of languages satisfy the equation (ZG), then the same is true of their closure
under Boolean operations. This is in fact true for arbitrary equations over
monoids: we have shown this result as 
\cref{clm:closure} in
Appendix~\ref{sec:algebra}.

We are now ready to show \cref{lem:almost_comm_satisfies_ZG_words}:

\begin{proof}[Proof of \cref{lem:almost_comm_satisfies_ZG_words}]
  It suffices to show the claim when $L$ is regular-commutative or
  virtually-singleton, as we can then conclude by \cref{clm:closure}.
  Let $\mu$ be the syntactic morphism of~$L$.

  For the first case, assume that \(\LL\) is a regular-commutative language. In this case, we
  will show that the monoid $V$ is in fact commutative, which implies the
  equation (ZG).

  Let \(v,w\in V\) be two elements of the vertical monoid: we want to show that
  $vw = wv$, where the product notation denotes the composition law of~$V$
  (which is also written $\odotVV$). Let \(C,D\) be
  two $\Sigma$-contexts mapped respectively to \(v\) and \(w\) by \(\mu\):
  this uses the fact that the syntactic morphism is surjective.
  We want to show that \(\mu(C(D)) = \mu(D(C))\): by the minimality
  condition on the syntactic forest algebra, we can do this by establishing
  that,
  For every context $E$ and forest $F$, we have 
  $E(C(D(F))) \in L$ iff $E(D(C(F))) \in L$.
  Now, the two forests $E(C(D(F)))$ and $E(D(C(F)))$ have the same Parikh image,
  so by definition of regular-commutative languages, either both belong to~$L$
  or none belong to~$L$. Thus, by minimality, we have $\mu(C(D)) = \mu(D(C))$,
  so $vw = wv$. Thus, we have established that $V$ is commutative; in particular
  it satisfies the equation (ZG).

  For the second case, assume that \(\LL\) is a virtually-singleton language. In this case, we
  will show that $V$ is in fact \emph{nilpotent}, by which we mean that it has
  at most two idempotent elements (i.e., elements $x$ such that $xx = x$): the
  neutral element, and (potentially) a \emph{zero}, meaning an element $z$ such
  that $xz = zx = z$ for all elements~$x$.
  Note that when a monoid contains a zero then it is unique.
  Formally, nilpotent monoids are
  those obtained by adding a neutral element to a nilpotent semigroup: cf
  \cite{amarilli2021dynamic,Straubing82}. In nilpotent monoids, group elements
  (of the form $x^{\omega+1}$) are either the neutral element or the zero, and
  so they are central (i.e., they commute with all other elements), so that
  (ZG) is satisfied.

  In the definition of $L$, call \(\Sigma'\) be the subalphabet over which we
  are projecting, and call 
  \(F'\) be the $\Sigma'$-tree to which we must be equal after projection (see
  \cref{def:virtsing}).
  Let \(v\in V\): we want to show that $v^{\omega+1}$ is either the neutral
  element or a zero. By surjectivity, let $C$ be a context mapped to \(v\) by
  \(\mu\). We distinguish two cases depending on whether $C$ contains a letter
  in~$\Sigma'$ or not.

  If \(C\) has no letter in \(\Sigma'\), then it is clear that it is equivalent
  to the neutral context (i.e., the forest with a singleton root labeled
  $\square$).
  Indeed, for any context $E$ and forest $F$, it is clear that
  $E(C(F)) \in L$ iff $E(F) \in L$, so that by minimality we must have that
  $\mu(C)$ is equal to the image by~$\mu$ of the empty context, which is the
  neutral element of~$V$ by definition of a morphism.

  If \(C\) has a letter in \(\Sigma'\), then let \(n'\) be strictly greater than
  the size of \(F'\).
  Then for every context \(E\) and forest \(F\), we claim that \(E(C^{n'}(F))\) is
  not in~$L$,
  where $C^{n'}$ denotes repeated application of the context $C$. Indeed, 
  such a forest contains at least $n'$ letters
  of~$\Sigma'$, so its projection to~$\Sigma'$ cannot be~$F'$. Let us show that
  $\mu(C^{n'})$ is a zero. For any element $v' \in V$, letting $C'$ be a context
  such that $\mu(C') = v'$ by minimality, we have $\mu(C^{n'}) v' = \mu(C^{n'})
  \mu(C') = \mu(C^{n'} C')$.
  But, for every context $E$ and forest $F$, it is again
  the case that $E (C^{n'}(C'(F)))$ is not in the language, so by minimality $C^{n'}$
  and $C^{n'} C'$ are mapped to the same element by~$\mu$, and thus $\mu(C^{n'}) v' =
  \mu(C^{n'})$. Symmetrically, one can show that $v' \mu(C^{n'}) = \mu(C^{n'})$. Thus,
  $\mu(C^{n'})$ is a zero in~$V$; and zeroes are unique, so it is the zero of~$V$.
  We have shown that $v^{n'}$ is the zero of~$V$.
  Let $k$ be such that $k \omega > n'$. Then $v^{\omega+1} = v^{k \omega + 1}$,
  thanks to idempotence. But $v^{k \omega + 1}$ contains $v^{n'}$ as a factor, so
  as it is the zero of~$V$, so is $v^{k \omega + 1}$, hence so is
  $v^{\omega+1}$. This implies that $v^{\omega+1}$ is central, and concludes.
\end{proof}

\subsubsection{Backward Direction}
The challenging direction to prove \cref{thm:almost-comm_zg} is the backward
direction:
let $L$ be a regular language whose syntactic forest algebra is in ZG, and let
us show that $L$ is almost-commutative.

The proofs in this section will use equations on ZG forest algebras shown in
Appendix~\ref{apx:equations}: of course we will use
(ZG) from \cref{def:zg} on the vertical monoid (by the definition of ZG forest
algebras),
but we will also use the ZG equation on the horizontal monoid (i.e., equation
(ZGh) from \cref{clm:zgh}),
and we will use (DISTv) and (DISTh)
from \cref{clm:distomega},
(OUTh) and (OUTv)
from \cref{lem:OUTh_OUTv},
(IDv)
from \cref{lem:IDv},
and (FLAT) from \cref{lem:flat_eq}.

We will show that, when a $\Sigma$-forest is mapped to an idempotent by
a morphism $\mu$ into a ZG forest algebra, then we can put the forest
into a normal form that has the same image by the morphism.

\begin{lemma}\label{lem:normal_form_idempotent}
  Let $\mu$ be a morphism from $\Sigma$-forests and $\Sigma$-contexts to a
  forest algebra $(V,H)$ in (ZG), and let $m$ be the idempotent power of~$V$.

  Let \(F\) be a forest mapped to an idempotent of \(H\) by~$\mu$, and let $a_1,
  \ldots, a_k$ be the distinct letters of~$\Sigma$ that occur in~$F$.
  Define the forest
  \[ \Xi_F = a_{1}^{n}+ \cdots + a_{k}^{m} \]
  where $a_i^m$ denotes the line tree with $m$ nodes labeled $a_i$, each node
  having exactly one child (except the last).
  Then we have:
  \[ \mu(F) = \mu(\Xi_F).  \]
\end{lemma}

Intuitively, this lemma tells us that for morphisms to ZG forest algebras, in
particular for the syntactic morphism, all forests mapped to an idempotent are
indistinguishable if they contain the same set of letters. Let us prove the
lemma:

\begin{proof}
  Let us show by induction on \(F\) that, for every forest~$F$, we have:
  \[\omega\cdot \mu(F)= \mu(\Xi_F).\]
  Showing this suffices to conclude, because the lemma assumes that $F$ is
  mapped to an idempotent, so that we have \(\mu(F) = \omega\cdot \mu(F)\).

  If \(F\) is empty then \(F=\Xi_{F}\) and thus \(\omega\cdot \mu(F) = \mu(F)=
  \mu(\Xi_{F})\).

  If \(F = F_{1}+ F_{2}\), then let \(b_{1},\ldots,b_{k_{1}}\) and \(c_{1},\ldots,c_{k_{2}}\) be the letters in \(F_{1}\) and \(F_{2}\).
  We have:
  \begin{align*}
    \omega\cdot \mu(F) &= \omega\cdot \mu(F_{1}) + \omega\cdot \mu(F_{2})  \tag{by (DISTh)} \\
            &= \mu(\Xi_{F_{1}}) + \mu(\Xi_{F_{2}}) \tag{by induction hypothesis} \\
            &= \mu(b_{1}^{m}) + \cdots + \mu(b_{k_{1}}^{m}) + \mu(c_{1}^{m}) + \cdots + \mu(c_{k_{2}}^{m}).
  \end{align*}
  Each of the term in the sum is an idempotent of \(V\), so by (IDv) in~\cref{lem:IDv} it is also an idempotent of \(H\), hence we can apply (ZGh)
  to commute them.
  Hence we can put them in any order and use idempotency to obtain only one copy of each letter.
  We conclude because the set of letters occurring in \(F\) is clearly the union
  of the letters of those occurring in \(F_{1}\) and of those occurring in \(F_{2}\).

  If \(F= a_{\square}(G)\), then let \(b_{1},\cdots,b_{k}\) be the letters in \(G\).
  We have:
  \begin{align*}
    \omega\cdot \mu(F) &= \omega \cdot  (\mu(a_{\square})\odot\mu(G)) \\
            &= \mu(a_{\square})^{\omega}\odot\epsilon + \omega\cdot \mu(G) \tag{by (FLAT)} \\
            &= \mu(a^{m}) + \mu(\Xi_{G}) \tag{by induction hypothesis} \\
            &=  \mu(a^{m}) +  \mu(b_{1}^{m}) +\cdots +   \mu(b_{k}^{m})
  \end{align*}
  We conclude exactly like in the previous case.
\end{proof}

Let $(V,H)$ be a forest algebra, let $\mu$ be a morphism to~$(V,H)$, and let $F$
be a forest.
We say that a context \(C\) is an \emph{idempotent factor} of \(F\) if
\begin{itemize}
  \item it is non-empty,
  \item \(\mu(C)\) is an idempotent of \(V\),
  \item there exists a context \(D\) and a forest \(G\) such that \(F = D(C(G))\).
\end{itemize}

We will now show that every sufficiently large forest must contain an idempotent
factor:

\begin{fact}\label{fact:idempotent_in_tree}
  Let $\mu$ be a morphism to a forest algebra $(V,H)$ and let $F$ be a forest.
  If we have $|F| > |V|^{5|V|^{6|V|}}$,
  then it is possible to find an idempotent factor in \(F\).
\end{fact}

This is a more complicated analogue of finding idempotent factors in
long words (see \cite{jecker_ramsey_2021} for fine bounds on the word length):

\begin{lemma}[\cite{jecker_ramsey_2021}, Theorem 1]
  \label{lem:idemword}
  Let \(M\) be a finite monoid and let $\mu$ be a morphism to~$M$.
  For any word \(w\in M^{q}\) with \(q\geq |M|^{5|M|} \),
there is a subword of \(w\) that is mapped to an idempotent.
\end{lemma}

Intuitively, we will use this lemma to prove \cref{fact:idempotent_in_tree}
because a sufficiently large forest is either
``wide'' (i.e., contains a large sibling set) or ``deep'' (i.e., contains a long
path). Let us prove our claim:

\begin{proof}[Proof of \cref{fact:idempotent_in_tree}]
For any forest $F$, one of the following three cases must occur:
\begin{itemize}
  \item There is a set of trees that are siblings of size greater than \(|V|^{5|V|}\) (that is to say that there
        is a node with more than \(|V|^{5|V|}\) children, or there are more than \(|V|^{5|V|}\) roots).
        In this case, we denote them by \(G_{1},\ldots,G_{q}\), with $q > |V|^{5|V|}$.
        Let \(C_{i}\) be the context \(G_{i}+\square\).
        We construct a word \(w\in V^{q}\) by setting \(w_{i} = \mu(C_{i})\).
        Because \(q\) is big enough, by \cref{lem:idemword} on the horizontal
        monoid,
        this word contains an idempotent \(\mu(C_{i}\cdots C_{j})\).
        Let \(D\) be the context that consists of \(F\) with \(G_{i}+\cdots+ G_{j}\) identified in a single node \(\square\),
        and let \(C = C_{i}(\cdots  (C_{j}))\).
        Thus \(F= D(C(\epsilon))\) and \(\mu(C)\) is an idempotent.
  \item There is a path from the root to a leaf of length greater than \(|V|^{5|V|}\).
        In this case, let 
        $u_1, \ldots, u_q$ be the nodes along this path and let
        \(a_{1},\ldots,a_{q}\) be their respective labels.
        For \(1\leq i \leq q\), let \(L_{i}\) (resp., \(R_{i}\)) be the forest with
        the left (resp., right) siblings of \(u_{i+1}\), and
        define \(C_{i} = L_{i}+a_{i}(\square)+R_{i}\).
        With these definitions, we have that \(F = C_{1}(C_{2}(\cdots (C_{q}(\epsilon))))\).
        We construct a word \(w\in V^{q}\) by setting \(w_{i} = \mu(C_{i})\).
        Because \(q\) is big enough,
        by \cref{lem:idemword} on the vertical monoid,
        we can find an idempotent \(\mu(C_{i}\cdots C_{j})\).
        Let \(D \coloneq C_{1}(\cdots (C_{i-1}))\), let
        \(C \coloneq C_{i}(\cdots (C_{j}))\), and let
        \(G\coloneq C_{j+1}(\cdots (C_{q}(\epsilon)))\).
        We have \(F = D(C(G))\) and \(\mu(C)\) is an idempotent.
  \item Every node has less than \(|V|^{5|V|}\) children and every path from the root to a leaf is of length less than \(|V|^{5|V|}\).
        In this case, \(F\) has size less than
        \(|V|^{5|V|\cdot |V|^{5|V|}}\), so
        such forests are excluded by the bound assumed in the result statement.
        \qedhere
\end{itemize}
\end{proof}

We are now ready to show that, for morphisms to a ZG forest algebra, every
forest~$F$ can be put into a \emph{normal form} which has the same image by the
morphism as~$F$. Intuitively, the normal form of a forest
consists of a forest on the alphabet of the \emph{rare letters}, i.e., those
letters with a ``small'' number of occurrences in~$F$, along with line trees
counting the number of occurrences of the \emph{frequent letters} modulo the
idempotent power of the vertical monoid. Formally:

\begin{figure}
  \begin{subfigure}[t]{0.38\linewidth}
    \centering
    \includegraphics[scale=0.75]{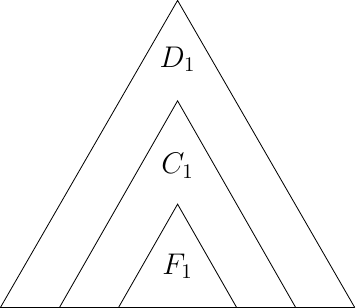}
    \caption{Identification of an idempotent $C_1$.}
    \label{subfig:fishing:find}
  \end{subfigure}
  {\unskip\ \vrule\ }
  \begin{subfigure}[t]{0.52\linewidth}
    \centering
    \includegraphics[scale=0.75]{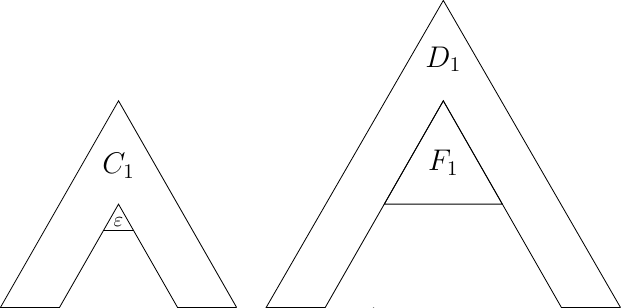}
    \caption{Extraction of~$C_1$ with (OUTh) and (OUTv).}
    \label{subfig:fishing:extract}
  \end{subfigure}

  \hrule

  \begin{subfigure}[t]{\linewidth}
  \vspace{\baselineskip}
    \centering
    \includegraphics[scale=0.75]{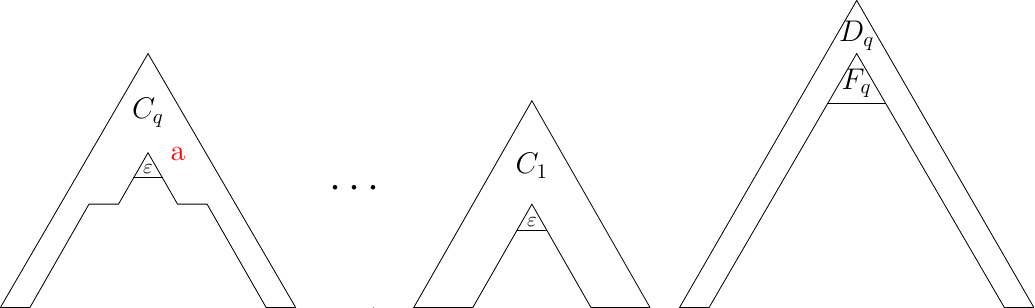}
    \caption{Repetition until an idempotent with a letter $a$ in \(\SigmaF\) is found.}
    \label{subfig:fishing:frequent}
  \end{subfigure}

  \hrule

  \begin{subfigure}[t]{0.45\linewidth}
  \vspace{\baselineskip}
    \centering
    \includegraphics[scale=0.75]{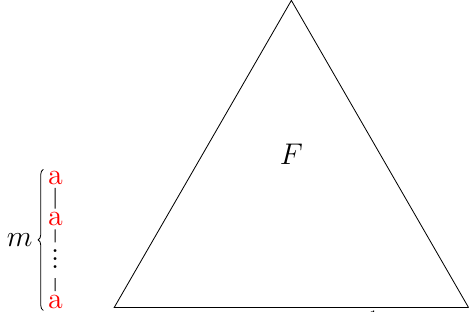}
    \caption{Extraction of many \(a\)'s with \cref{lem:normal_form_idempotent}, and reconstruction of \(F\).}
    \label{subfig:fishing:separate}
  \end{subfigure}
  {\unskip\ \vrule\ }
  \begin{subfigure}[t]{0.45\linewidth}
  \vspace{\baselineskip}
    \centering
    \includegraphics[scale=0.75]{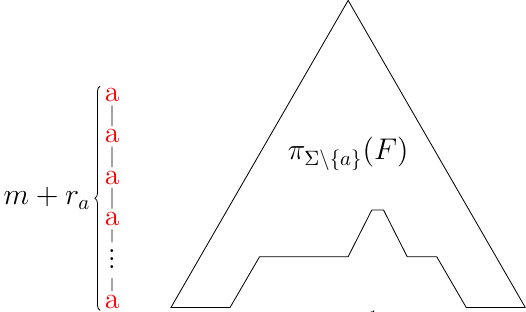}
    \caption{Extraction of every \(a\) from \(F\) with (OUTh) and (OUTv).}
    \label{subfig:fishing:fishing}
  \end{subfigure}

  \hrule

  \begin{subfigure}[t]{\linewidth}
  \vspace{\baselineskip}
    \centering
    \includegraphics[scale=0.75]{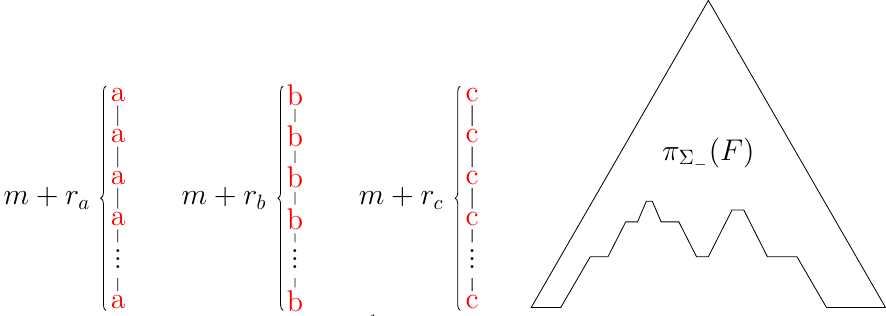}
    \caption{Extraction of every letter in \(\SigmaF\) from \(f\).}
    \label{subfig:fishing:induct}
  \end{subfigure}

  \caption{Proof of \cref{lem:normal_form_ZG}}
  \label{fig:fishing}
\end{figure}

\begin{lemma}\label{lem:normal_form_ZG}
  Write the letters of the alphabet $\Sigma$ as $a_1, \ldots, a_k$.
  Let $(V,H)$ be a forest algebra in ZG, and let
  $m$ be the idempotent power
  of~$V$.
  Let \(N =|V|^{5|V|^{6|V|}} \).
  Let $\mu$ be a morphism into~$(V,H)$, and 
  let \(F\) be a forest.
  Recalling the definition of the Parikh image (\cref{def:regcom}), let
  \((q_{a_1},\ldots,q_{a_k})\) be the Parikh image of \(F\), and let
  \(r_{a_i}=q_{a_i} \text{ mod } m\).
  We partition \(\Sigma=\SigmaR\cup \SigmaF\) with
  \begin{align*}
    \SigmaR &= \{ a\ |\ q_{a}< N\}, \text{~called the \emph{rare letters}}; \\
    \SigmaF &= \{ a\ |\ q_{a}\geq N\}, \text{~called the \emph{frequent letters}}.
  \end{align*}

  The image by~$\mu$ of~$F$ is then equal to that of the following forest:
  \[ \Phi_F \coloneq \sum_{a\in \SigmaF} a^{m+r_a} + \pi_{\SigmaR}(F) \]
  where $a^{m+r_a}$ denotes a line tree with $m+r_a$ nodes labeled~$a$, and
  $\pi_{\SigmaR}$ denotes the projection to the subalphabet~$\SigmaR$ as in
  \cref{def:virtsing}.
\end{lemma}

\begin{proof}
  The proof is graphically represented in \cref{fig:fishing}.
  We proceed by induction on the number of frequent letters, i.e., of letters in
  \(\SigmaF\).
  If \(\SigmaF=\emptyset\), then \(\Phi_{F} = F\) and they have the same images under \(\mu\).
  Now assume that \(\SigmaF\) has at least one letter. The proof proceeds in several
  steps. First, we want to identify an idempotent factor of~$F$ that contains a
  letter $a$ in~$\SigmaF$. Second, we want to argue that adding a sibling tree
  $a^n(\epsilon)$ to~$F$ does not change its image by~$\mu$. Third, we want to use this
  sibling tree to collect all occurrences of~$a$ in~$F$, and argue that $F$ has
  the same image by~$\mu$ as $a^{n+r_a} + \pi_{\Sigma \setminus \{a\}}(F)$: this
  is the analogue of a result on the ZG-congruence on words \cite[Claim
  3.9]{amarilli2023locality}. Repeating this argument for each letter
  of~$\SigmaF$
  yields the claim.

  Let us first perform the first step, where we find an idempotent factor in \(F\)
  containing a letter in \(\SigmaF\).
  Let \(F_{0} = F\).
  Assume we have constructed a sequence of forests \(F_{0},\ldots,F_{i}\) and a sequence of contexts \(C_{1},\ldots, C_{i}\),
  such that every \(F_{j}\) has size greater than \(N\).
  Thanks to that, we can apply \cref{fact:idempotent_in_tree} to \(F_{i}\).
  This give the decomposition in \cref{subfig:fishing:find}.
  Hence we can write \(F_{i}= D_{i+1}(C_{i+1}(F'_{i+1}))\) with \(\mu(C_{i+1})\) an idempotent of \(V\).
  Let \(F_{i+1} \coloneq D_{i+1}(F'_{i+1})\).
  We claim that \(\mu(F_{i}) = \mu(C_{i+1}(\epsilon) + F_{i+1})\),
  as represented in \cref{subfig:fishing:extract}.
  Indeed, with \(v = \mu(D_{i+1})\), \(w=\mu(C_{i+1})\) an idempotent and \(h = \mu(F'_{i+1})\):
  \begin{align*}
    \mu(F_{i}) &= \mu(D_{i+1}(C_{i+1}(F'_{i+1}))) \\
             &= v\odot (w^{\omega} \odot h) \\
             &= v \odot (w^{\omega}\odot\epsilon + h) \tag{by (OUTv)} \\
             &=v \odot (\omega \cdot (w^{\omega}\odot \epsilon) + h) \tag{by (IDv)} \\
             &= v\odot h + \omega \cdot (w^{\omega}\odot\epsilon) \tag{by (OUTh)} \\
             &= v\odot h + (w^{\omega}\odotVH \epsilon) \tag{by (IDv)} \\
             &=  (w^{\omega}\odot\epsilon) + v\odot h \tag{by (ZGh)} \\
             &= \mu(C_{i+1}(\epsilon) + F_{i+1})
  \end{align*}
  For the use of (OUTh), we are using \(v\odot(\square + h)\) as the context and \(w^{\omega} \odot\epsilon\) as the forest.
  If \(C_{i+1}\) contains a letter in \(\SigmaF\), the construction terminates and we
  continue with the next paragraph.
  Otherwise, \(F_{i+1}\) and \(F\) have the same number of occurrences of every
  letter of \(\SigmaF\), which is non empty.
  This implies that \(F_{i+1}\) has size greater than \(N\), and we can repeat the process.
  The size of \(F_{i}\) decreases at each step, so the construction must terminate.

  Let \(F_{0},\ldots,F_{q}\) and \(C_{1},\ldots,C_{q}\) be the obtained
  sequences, and \(a\) be the letter in \(\SigmaF\) that appears in \(C_{q}\).
  In the end we obtain that, as represented in \cref{subfig:fishing:frequent}:
  \begin{equation} 
    \label{eqn:exploding}
    \mu(F) = \mu(C_{q}(\epsilon)) + \cdots + \mu(C_{1}(\epsilon)) + \mu(F_{q}).
  \end{equation}
  In this equation, $C_q$ is an idempotent context containing $a$, so we have
  performed the first step.

  \medskip

  Let us now show the second step, where we argue that we can add a sibling tree
  $a^m$ to~$F$ without changing the image by~$\mu$. We will show this first
  by adding the new sibling tree to $C_q(\epsilon)$.
  Consider the forest \(F' = a^{m} + C_{q}(\epsilon)\).
  The contexts \(\mu(a^{m}_{\square})\) and \(\mu(C_{q})\) are vertical idempotents, and therefore, by (IDv), both \(\mu(a^{m})\) and \(\mu(C_{q}(\epsilon))\) are horizontal idempotents.
  So by (DISTh), \( \mu(g)\) is also an idempotent.
  Let \(\Xi_{F'}\) and \(\Xi_{C_{q}(\epsilon)}\)  be the normal-form forests as defined as in the statement of \cref{lem:normal_form_idempotent}.
  We know that \(F'\) and \(C_{q}(\epsilon)\) contain the same letters (because the
  letter $a$ occurs in~$C_q$), and therefore \(\Xi_{F'} = \Xi_{C_{q}(\epsilon)}\).
  So by \cref{lem:normal_form_idempotent}:
  \[ \mu(F') = \mu(\Xi_{F'}) = \mu(\Xi_{C_{q}(\epsilon)}) = \mu(C_{q}(\epsilon)).  \]
  This gives, as represented in \cref{subfig:fishing:separate}:
  \begin{align*}
    \mu(F) &=\mu(a^{m}) + \mu(C_{q}(\epsilon)) + \cdots + \mu(C_{1}(\epsilon)) + \mu(F_{q}) \\
         &= \mu(A^{m}) + \mu(F)
  \end{align*}
  where in the last equality we have used \cref{eqn:exploding} to reconstruct
  $F$. We have completed the second step of the proof.

  \medskip

  We now proceed with the third (and last) step of the proof, where we use the extra factor to
  collect the occurrences of~$a$ in~$F$.
  Recall that $q_a = |F|_a$ denotes the number of~$a$'s in~$F$.
  We build a sequence of forests \(G_{0},\ldots, G_{p}\)
  such that
  \begin{itemize}
    \item \(G_{0}=F\);
    \item for every \(0\leq i\leq q_a\), we have \(\mu(a^{m+i} + G_{i}) = \mu(F)\);
    \item for every \(0\leq i\leq q_a\), we have $|G_i|_a = q_a - i$, so in particular $|G_{q_a}|_a = 0$.
  \end{itemize}

  We show how to construct \(G_{i+1}\) 
  from \(G_{i}\).
  If $G_i$ no longer contains any~$a$, then we have obtained $G_{q_a}$ and we have
  finished, so
  assume that there still is an \(a\) in \(G_{i}\). In this case, we can write \(G_{i} = C'_{i}(a_{\square}(G_{i}'))\).
  Let \(G_{i+1} = C_{i}'(G_{i}')\).
  Let \(v = \mu(a_{\square})\), \(w_{i}= \mu(C'_{i})\) and \(h_{i}= \mu(G_{i}')\).
  The value \(v^{\omega+i}\odot\epsilon\) can be written, thanks to \cref{lem:IDv}, as \((\omega+1) \cdot (v^{\omega+i}\odot\epsilon)\).
  So we can apply (OUTh) with \(v^{\omega+i}\odot\epsilon\).
  We have that:
  \begin{align*}
    \mu(F) &= \mu(a^{m+i} + G_{i}) \\
         &= v^{\omega+i}\odot\epsilon + w_{i}\odot(v\odot h_{i}) \\
         &=  w_{i}\odot(v^{\omega+i}\odot\epsilon + v\odot h_{i}) \tag{by (OUTh)} \\
         &=  w_{i}\odot(v^{\omega+i}\odot(v\odot h_{i})) \tag{by (OUTv)} \\
         &=  w_{i}\odot(v^{\omega+i+1}\odot h_{i})  \\
         &=  w_{i}\odot(v^{\omega+i+1}\odot\epsilon + h_{i}) \tag{by (OUTv)} \\
         &= v^{\omega+i+1}\odot\epsilon + w_{i}\odot h_{i} \tag{by (OUTh)} \\
         &=  \mu(a^{m+i+1} + G_{i+1})
  \end{align*}

  For the first use of (OUTh), we are using \(w_{i}\odot(\square \oplus (v\odot h_{i}))\) as the context and \( v^{\omega+i}\odot \epsilon\) as the forest.
  For the second use of (OUTh), we are using the context \( w_{i}\odot(\square \oplus  h_{i}))\) and the forest \( v^{\omega+i+1}\odot \epsilon\).
  In both cases, we are using (IDv) implicitly to argue that the forests are idempotent in \(H\).

  We have indeed that \(G_{i+1}\) has one less \(a\) than \(G_{i}\).
  In the end, thanks to the conservation of the number of \(a\)'s, we have that
  \begin{align*}
    \mu(F) &= \mu(a^{m+q_{a}} + G_{q_a}) \\
     &= \mu(a^{m+r_{a}} + G_{q_a})
  \end{align*}
  with \(G_{q_a}\) being the projection of \(F\) on \(\Sigma\setminus \{a\}\)
  and \(q_a\) is the
  number of \(a\)'s in \(F\).
  The situation is represented in \cref{subfig:fishing:fishing}.
  The last equality comes from the fact that \(\mu(a^{m}) = \mu(a^{2m})\), thus
  we can substract \(m\) to \(q_{a}\) until we reach \(r_{a} = q_{a} \text{ mod } m\). This concludes the third step.

  \medskip
  To finish the proof, we use the induction hypothesis on \(G_{q_a}\), that has
  one less letter in \(\SigmaF\) than \(F\).
  We obtain the figure in \cref{subfig:fishing:induct}.
\end{proof}

We can finally conclude the proof of \cref{thm:almost-comm_zg}:

\begin{proof}[Proof of~\cref{thm:almost-comm_zg}]
  By \cref{lem:almost_comm_satisfies_ZG_words}, we have the left-to-right
  implication, so we prove the other implication: 
  let \(\LL\) be a language recognised by a morphism $\mu$ to a forest algebra
  $(V, H)$ in ZG.
  Let \(N =|V|^{5|V|^{6|V|}}\) and let \(m\) be the idempotent power of \(V\).

  For any $\Sigma$-forest $E$, 
  remember the definition of the normal form forest \(\Phi_{E}\) from the statement of \cref{lem:normal_form_ZG}.
  We call the set \(\SigmaR\) the set of \emph{rare} letters in \(E\) and
  \(\SigmaF\) the
  set of \emph{frequent} letters in \(E\), and write them \(\SigmaR(E)\) and
  \(\SigmaF(E)\)

  We then define an equivalence relation on $\Sigma$-forests as:
  \[
    F\sim G \text{ iff }
    \left\{\begin{array}{l}
      \SigmaR(F) = \SigmaR(G) \text{ and } \SigmaF(F)=\SigmaF(G) \\
            \pi_{\SigmaR(F)}(F) = \pi_{\SigmaR(G)}(G)  \\
           \forall a\in \SigmaF(F), |F|_{a} \equiv |G|_{a} \text{ mod } m
    \end{array}\right. . \]

  One can check that, by definition, \(F\sim G\) then \(\Phi_{F}=\Phi_{G}\). Thus, 
  by \cref{lem:normal_form_ZG}, if  \(F\sim G\)
  then \(\mu(F)=\mu(G)\).
  Moreover, there are finitely many equivalence classes of \(\sim\).
  Indeed, trees of the form \(\pi_{\SigmaR(F)}(F)\) have at most \(|\Sigma|\cdot N\) letters.
  From these two facts we deduce that \(\LL\) is a finite union of equivalence classes of \(\sim\).
  All is left to do is to show that an equivalence class \(X\) of \(\sim\) is an almost-commutative language.

  Let us now fix the equivalence class $X$: choose disjoint subalphabets
  \(\SigmaR\) and \(\SigmaF\) such that \(\Sigma= \SigmaR\cup \SigmaF\), let \(E\) be a forest over
  \(\SigmaR\) having at most $|\Sigma| \cdot N$ letters,
  and let \(r\in \{0,\ldots,n-1\}^{\SigmaF}\) such that
  \(X\) is the set of forests whose rare letters are \(\SigmaR\), whose frequent
  letters are \(\SigmaF\), whose projection over \(\SigmaR\) is \(E\) and
  where for every \(a\in \SigmaF\), the number of \(a\)'s is congruent to \(r_{a}\) modulo \(m\).
  We define \(\LL_{1}\) to be the virtually-singleton language of forests whose projection over \(\SigmaR\) is \(E\).
  Let \(S\) be the 
  set of the vectors \(x\) such that
  \begin{itemize}
    \item \(x_{a} \geq N\) and \(x_{a} \equiv r_{a}\) modulo \(n\) for every \(a\in \SigmaF\),
    \item \(x_{a} = |E|_{a}\) for every \(a\in \SigmaR\).
  \end{itemize}
  We define \(\LL_{2}\) to be the commutative language of forests whose Parikh
  image is in \(S\), which is easily seen to be regular because \(S\) is
  ultimately periodic.
  We have that
  \[ X = \LL_{1} \cap \LL_{2}, \]
  proving that it is almost-commutative.

  Thus, $L$ is a finite union of equivalence classes which are all
  almost-commutative languages, so $L$ is almost-commutative. This concludes the
  proof.
\end{proof}

\end{toappendix}

Proving this algebraic characterization is the main technical challenge of the
paper.
Let us sketch the proof of \cref{thm:almost-comm_zg}, with the details given in
Appendix~\ref{apx:zgproof}:

\begin{proofsketch}
  The easy direction is to show that almost-commutative languages have a
  syntactic forest algebra in ZG. We first show this for commutative languages
  (whose vertical monoids are unsurprisingly commutative) and for
  virtually-singleton languages (whose vertical monoids are nilpotent, i.e.,
  are in the class MNil of~\cite{amarilli2021dynamic,Straubing82}). We conclude
  because satisfying the ZG equation is preserved under Boolean operations.

  The hard direction is to show that any regular language $L$ whose
  syntactic forest algebra is in ZG is 
  almost-commutative, i.e., can be expressed as a finite Boolean combination of
  virtually-singleton and regular-commutative languages. For this, we show that
  morphisms to a ZG forest algebra must be determined by the
  following information on the input forest: which letters are \emph{rare}
  (i.e., occur a constant number of times) and which are \emph{frequent}; how
  many times the frequent letters appear modulo the idempotent power of the
  monoid; and what is the projection of the forest on the rare letters. These
  conditions amount to an almost-commutative language, and are the analogue for
  trees of results on ZG languages \cite{amarilli2021dynamic}. The proof is
  technical, because we must show how the ZG equation on the vertical monoid
  implies that every sufficiently frequent letter commutes
  both vertically and horizontally.
\end{proofsketch}

\subparagraph*{Hardness for syntactic forest algebras outside ZG.}
To show our conditional hardness result (\cref{thm:lb}), what remains is to
show that the dynamic membership problem is hard for regular
languages with a neutral letter and whose syntactic forest algebra is not in ZG:

\begin{toappendix}
  \subsection{Proof of \cref{prp:nonzghard}}
\end{toappendix}

\begin{propositionrep}
  \label{prp:nonzghard}
  Let $L$ be a regular forest language, and assume that it has a neutral letter
  and that its syntactic forest
  algebra is not in ZG. Subject to \cref{conj:prefix_uone}, the dynamic membership
  problem for~$L$ cannot be solved in constant time per update.
\end{propositionrep}

\begin{proofsketch}
  We reduce from the case of words: from two contexts $v$ and
  $w^{\omega+1}$ witnessing that the ZG equation does not hold, we consider
  forests formed by the vertical composition of a sequence of
  contexts which can be either $v$ or $w^{\omega+1}$, with a suitable context at
  the beginning and end. Let $L'$ be the word language of those sequences of
  contexts giving a forest in~$L$: we show that the syntactic monoid of~$L'$
  is not in ZG, and as $L'$ features a neutral letter, we conclude by~\cite{amarilli2021dynamic} that $L'$ does not enjoy $O(1)$
  dynamic membership assuming \cref{conj:prefix_uone}. Further, dynamic
  membership to~$L'$ can be achieved using a data structure for the same problem
  for~$L$ on the forests that we constructed, so hardness also holds for~$L$.
\end{proofsketch}

\begin{proof}
  Let \((V,H)\) be the syntactic forest algebra of \(L\),
  we know that \(V\) is not in ZG.
  Hence, let $v, w \in V$ be such that $v^{\omega+1} w \neq w v^{\omega+1}$.
  By surjectivity, let $C$ and $D$ be contexts achieving $v$ and~$w$, so that
  $\mu(C^{m+1}(D)) \neq \mu(D(C^{m+1}))$ for \(m\) any multiple of the idempotent power of \(V\).
  By the minimality property of the syntactic forest algebra,
  there exists a context $E$ and a
  forest $F$ that distinguishes $C^{m+1}(D)$ and $D(C^{m+1})$, i.e., precisely one of
  $E(C^{m+1}(D(F)))$ and $E(D(C^{m+1}(F)))$ is in~$L$.

  We will construct a language \(L'\) of words over the alphabet
  \(\Sigma'=\{\square,C,D\}\) enjoying the following properties: (1.) $L'$ will
  feature a neutral letter (namely $\square$), (2.) the dynamic membership problem
  for~$L'$ will not be in $O(1)$ subject to~\cref{conj:prefix_uone}, and (3.) a
  $O(1)$-algorithm for dynamic membership to~$L$ would yield such an algorithm
  for~$L'$.

  To define~$L'$, we can identify a word in \(\Sigma'\) with the context that consists of the
  concatenation of all letters of the word.
  The language \(L'\) is then defined as the set of words \(w\) such that
  \(E(w(F))\) is in \(L\): note that the letter \(\square\) of~$\Sigma'$ is
  indeed neutral for~$L'$, which establishes point (1.).

  To establish point (2.),
  we need to argue about properties of the syntactic monoid \(M\) of~$L'$.
  We have not defined the algebraic theory for word languages $L'$, though it is similar to
  the algebraic theory of forests.
  All we need to know is that there is a morphism \(\nu: (\Sigma')^{*} \rightarrow M\) that recognizes \(L'\), i.e.\
  there is \(P\subseteq M\) such that \(P=\nu(L')\).
  We claim that this monoid \(M\) is not in ZG.
  Let \(m\) be a multiple of both the idempotent powers of \(V\) and \(M\).
  We know that precisely one of \(E(C^{m+1}(D(F)))\) and \((E(D(C^{m+1}(F))))\) is in \(L\).
  By definition of \(L'\), it implies that
  \(\nu(C^{m+1}(D))\neq\nu(D(C^{m+1}))\),
  because precisely one of \(C^{m+1}(D)\) and \(D(C^{m+1})\) is in \(L'\).
  Thus the equation of ZG is violated in~$M$ by \(x=\nu(C)\) and \(y=\nu(D)\),
  as \(x^{\omega+1}y\neq yx^{\omega+1}\). 
  It follows from \cite{amarilli2021dynamic}
  that, subject to~\cref{conj:prefix_uone},  
  a language with a neutral letter (by (1.)) and whose syntactic monoid is not
  in ZG cannot have a \(O(1)\) dynamic membership problem. This establishes
  point (2.).

  To establish point (3.), 
  we denote by \(\gamma_{\square}\) the context \(C(D)\) in which all labels are replaced by the neutral letter.  Similarly, \(\gamma_{D}\) (resp., \(\gamma_{C}\)) is \(C(D)\) in which all
  labels in \(C\) (resp., in \(D\)) are replaced by the neutral letter.
  Note that the contexts \(\gamma_{\square}\), \(\gamma_{C}\), and \(\gamma_{D}\) all have the same shape.
  Now, if we want to maintain a word \(w\) of size \(n'\), we instantiate \(n'\) concatenated copies of \(\gamma_{\square}\) prefixed by \(E\) and suffixed by \(F\).
  We denote the result by \(T\). It is equivalent to \(E(F)\) at this stage.
  This preprocessing is in linear time, and the size $n$ of the constructed
  forest is in $\Theta(n')$.
  Then a substitution update in \(w\) to \(\square\) (resp., to \(C\), to \(D\))
  is translated into a constant number of relabeling updates in \(t\),
  so that \(\gamma_{\square}\) (resp., \(\gamma_{C}\), \(\gamma_{D}\)) appears at the right position.
  This implies that at every moment, \(T\) is equivalent to the forest \(E(w(F))\).
  Therefore, \(w\in L'\) is the same as \(T\in L\), establishing point (3.).

  Finally, we reason by contradiction in order to conclude.
  Assume that \(L\) has a \(O(1)\) dynamic membership algorithm.
  Then, by point (3.), \(L'\) also has a \(O(1)\) dynamic membership algorithm.
  This contradicts point (2.), which concludes.
\end{proof}

Note that \cref{prp:nonzghard} is where we use the assumption that there is a 
neutral letter. The result is not true otherwise, as we explain
below (and will discuss further in \cref{sec:extension}):

\begin{example}
  \label{exa:l0}
  Consider the language $L_0$ of forests over $\Sigma = \{a, b, c\}$ where there
is a node labeled $a$ whose next sibling is labeled~$b$. Membership to~$L_0$ can
be maintained in $O(1)$, like the language of words $\Sigma^* ab
\Sigma^*$. (Note that $c$ is not a neutral letter, because $ab$ is accepted but $acb$ is not.)
However, one can show that the syntactic
forest algebra of~$L_0$ is not in ZG. By contrast, adding a neutral letter
to~$L_0$
yields a language (with the same
syntactic forest algebra) with no $O(1)$ dynamic membership
algorithm under \cref{conj:prefix_uone}.
\end{example}

With \cref{prp:nonzghard} and \cref{thm:almost-comm_zg}, we can conclude the proof of \cref{thm:lb}:

\begin{proof}[Proof of \cref{thm:lb}]
  We already know that almost-commutative languages can be maintained efficiently
(\cref{thm:almost_cor_in_RAM}). Now, 
given a regular
forest language $L$ which is not almost-commutative and features a neutral letter,
we know by \cref{thm:almost-comm_zg}
  that its syntactic forest algebra is not in ZG, so we
conclude by \cref{prp:nonzghard}.
\end{proof}

\begin{toappendix}
  \vfill
\doclicenseThis
\end{toappendix}

\section{Conclusions and Future Work}
\label{sec:extension}
We have studied the problem of dynamic membership to fixed tree
languages under substitution updates. We have shown an $O(\log n / \log \log n)$
algorithm for arbitrary regular languages, and introduced the class of
almost-commutative languages for which dynamic membership can be decided in
$O(1)$. We have shown that, under the prefix-U1 conjecture, and if the
language is regular and features 
a neutral letter, then it must be almost-commutative to be
maintainable in constant time. Our work leaves many questions open which we
present below:
characterizing the
$O(1)$-maintainable languages without neutral letter, identifying languages with
complexity between $O(1)$ and $\Theta(\log n / \log \log n)$, and other open
directions.

\subparagraph*{Constant-time dynamic membership without neutral letters.}
Our conditional characterization of constant-time-maintainable regular languages
only holds in the presence of a neutral letter. In fact, without neutral letters,
it is not difficult to find non-almost-commutative forest languages with
constant-time dynamic membership, e.g., the language $L_0$ from \cref{exa:l0}.
There are more complex examples, e.g., dynamic membership in~$O(1)$ is
possible for the word language ``there is exactly one $a$ and exactly one $b$,
the $a$ is before the~$b$, and the distance between them is even'' (it is in
QLZG~\cite{amarilli2021dynamic}), and the same holds for the analogous forest
language.
We can even precompute ``structural information'' on forests of a more complex
nature than
the parity of the depths, e.g., to maintain in $O(1)$ ``there
is exactly one $a$ and exactly one $b$ and the path from $a$ to $b$ is a
downward path that takes only first-child edges''. We also expect other constant-time
tractable cases, e.g., ``there is exactly one node labeled $a$ and exactly
one node labeled $b$ and their least common ancestor is labeled $c$'' using
a linear-preprocessing constant-time data structure to answer least common
ancestor queries on the tree structure~\cite{harel1984fast}; or ``there is a node labeled
$a$ where the leaf reached via first-child edges is labeled~$b$''. These
tractable languages are not almost-commutative (and do not feature neutral
letters), and a natural direction for future work would be to characterize the
regular forest languages maintainable in $O(1)$ without the neutral letter
assumption.

\subparagraph*{Intermediate complexity for dynamic membership.}
We have explained in \cref{sec:prelim} that dynamic membership is in
$\Omega(\log n / \log \log n)$ for some regular forest languages, and
in \cref{sec:zgupper} that it is in $O(1)$ for almost-commutative
languages. One natural
question is to study intermediate complexity regimes. In the setting of
word languages, it was shown in~\cite{skovbjerg1997dynamic}
(and extended in~\cite{amarilli2021dynamic})
that
any aperiodic language~$L$ could be maintained in $O(\log \log n)$ per update.
This implies that some forest languages can be maintained with the same
complexity, e.g., the forests whose nodes in the prefix ordering form a word in an
aperiodic language~$L$.

The natural question is then to characterize which forest languages enjoy dynamic
membership between $O(1)$ and the general $\Theta(\log n /
\log \log n)$ bound.
We leave this
question open, but note a difference with the setting for words:
there are some aperiodic forest languages (i.e., both monoids of the syntactic
forest algebra are aperiodic) to which an $\Omega(\log n / \log \log n)$ lower
bound applies, e.g., the language for the existential marked ancestor problem
reviewed at the end of \cref{sec:prelim}. An intriguing question is whether
there is a dichotomy on 
regular forest languages, already in the aperiodic case, between those with 
$O(\log \log n)$ dynamic membership, and those with a
$\Omega(\log n / \log \log n)$ lower bound.

\subparagraph*{Other questions.}
Beyond relabeling updates
it would be natural to study the complexity of dynamic membership under
operations that can change the shape of the forest, e.g., leaf insertion and leaf deletion.
Another question concerns the support for more general queries than dynamic
membership, e.g., enumerating the answers to non-Boolean queries like
in~\cite{losemann2014mso,amarilli2018enumeration} (but with language-dependent
guarantees on the update time to improve over $O(\log n)$). Last, another
generalization of dynamic membership to forest languages is to consider the
same problem for context-free
languages, e.g., for Dyck languages (\cite[Proposition~1]{husfeldt1998hardness}),
or for visibly pushdown languages. Note that this is a different setting
from forest languages because substitution updates may change the shape of
the parse tree.

\vfill
\pagebreak

\bibliography{bib}

\end{document}